\documentclass[10pt,a4paper]{article}
\usepackage{amssymb}
\usepackage{amsmath}
\usepackage{amsthm}
\usepackage{latexsym}
\usepackage[dvips]{epsfig}
\usepackage{mathrsfs}
\usepackage{yhmath}

\theoremstyle{plain}
\newtheorem{theorem}{Theorem}
\newtheorem{proposition}[theorem]{Proposition}
\newtheorem{lemma}[theorem]{Lemma}
\newtheorem{corollary}[theorem]{Corollary}

\newtheorem*{main}{Theorem}
\newtheorem*{definition}{Definition}

\font\SYM=msbm10
\newcommand{\Real}{\mbox{\SYM R}}
\newcommand{\Complex}{\mbox{\SYM C}}

\newcommand{\Sphere}{\mbox{\SYM S}}

\font\tenscr=rsfs10 scaled1100
\font\sevenscr=rsfs7 
\font\fivescr=rsfs5 
\skewchar\tenscr='177
\skewchar\sevenscr='177
\skewchar\fivescr='177
\newfam\scrfam
\textfont\scrfam=\tenscr
\scriptfont\scrfam=\sevenscr
\scriptscriptfont\scrfam=\fivescr

\newcommand{\updn}[3]{#1^{#2}_{\phantom{#2}#3}}

\begin{document}


\title{\textbf{On the construction of a geometric invariant measuring
    the deviation from Kerr data}}

\author{{\Large Thomas B\"ackdahl} \thanks{E-mail address:
{\tt t.backdahl@qmul.ac.uk}} \\
\vspace{5mm}
{\Large Juan A. Valiente Kroon} \thanks{E-mail address:
{\tt j.a.valiente-kroon@qmul.ac.uk}}\\
School of Mathematical Sciences,\\
 Queen Mary University of London, \\
Mile End Road, London E1 4NS, UK.}

\maketitle

\begin{abstract}
This article contains a detailed and rigorous proof of the
construction of a geometric invariant for initial data sets for the
Einstein vacuum field equations. This geometric invariant vanishes if
and only if the initial data set corresponds to data for the Kerr
spacetime, and thus, it characterises this type of data.  The
construction presented is valid for boosted and non-boosted initial
data sets which are, in a sense, asymptotically Schwarzschildean. As a
preliminary step to the construction of the geometric invariant, an
analysis of a characterisation of the Kerr spacetime in terms of
Killing spinors is carried out. A space spinor split of the
(spacetime) Killing spinor equation is performed, to obtain a set of
three conditions ensuring the existence of a Killing spinor of the
development of the initial data set. In order to construct the
geometric invariant, we introduce the notion of approximate Killing
spinors. These spinors are symmetric valence 2 spinors intrinsic to
the initial hypersurface and satisfy a certain second order elliptic
equation ---the approximate Killing spinor equation. This equation
arises as the Euler-Lagrange equation of a non-negative integral
functional. This functional constitutes part of our geometric
invariant ---however, the whole functional does not come from a
variational principle. The asymptotic behaviour of solutions to the
approximate Killing spinor equation is studied and an existence
theorem is presented.
\end{abstract}

\section{Introduction}

The Kerr spacetime is, undoubtedly, one of the most important exact
solutions to the vacuum Einstein field equations \cite{Ker63}. It describes a
rotating stationary asymptotically flat black hole parametrised by its
mass $m$ and its specific angular momentum $a$. One of the
outstanding challenges of contemporary General Relativity is to obtain
a full understanding of the properties and the structure of the Kerr
spacetime, and of its standing in the space of solutions to the
Einstein field equations.

\medskip
There are a number of difficult conjectures and partial results
concerning the Kerr spacetime. In particular, it is widely expected to
be the only rotating stationary asymptotically flat black hole. This
conjecture has been proved if the spacetime is assumed to be analytic
($C^\omega$) --- see e.g. \cite{ChrCos08} and references
within. Recently, there has been progress in the case where the
spacetime is assumed to be only smooth ($C^\infty$) ---see
\cite{IonKla09a}. Moreover, it has been shown that a regular,
non-extremal stationary black hole solution of the Einstein vacuum
equations which is suitably close to a Kerr solution must be that Kerr
solution ---i.e. \emph{perturbative stability} among the class of
stationary solutions \cite{AleIonKla09}.

\medskip
Another of the conjectures concerning the Kerr spacetime is 
that it describes, in some
sense, the late time behaviour of a spacetime with dynamical (that is,
non-stationary) black holes ---this is sometimes known as the
\emph{establishment point of view of black holes},
cfr. \cite{Pen73}. A step in this direction is to obtain a proof of
the \emph{non-linear stability} of the Kerr spacetime ---this conjecture roughly 
states that the Cauchy problem for the vacuum Einstein field
equations with initial data for a black hole which is suitably close
to initial data for the Kerr spacetime gives rise to a spacetime with the same
global structure as Kerr and with suitable pointwise decay. Numerical
simulations support the conjectures described in this paragraph.

\medskip
A common feature in the problems mentioned in the previous paragraphs
is the need of having a precise formulation of what it means that a
certain spacetime is \emph{close} to the Kerr solution. Due to the
coordinate freedom in General Relativity it is, in general, difficult
to measure how much two spacetimes differ from each other. Statements
made in a particular choice of coordinates can be deceiving. In the
spirit of the geometrical nature of General Relativity, one would like
to make statements which are coordinate and gauge
independent. Invariant characterisations of spacetimes provide a way
of bridging this difficulty.

\medskip
Most analytical and numerical studies of the Einstein field equations
make use of a 3+1 decomposition of the equations and the
unknowns. In this context, the question of whether a given initial
data set for the Einstein field equations corresponds to data for the
Kerr spacetime arises naturally ---an initial data set will be said to
be data for the Kerr spacetime if its development is isometric to a
portion (or all) of the Kerr spacetime. A related issue arises when
discussing the (either analytical or numerical)  3+1 evolution of a
spacetime: do the leaves of the foliation approach, as a result of the
evolution, hypersurfaces of the Kerr spacetime?
In order to address these issues it is important to have a geometric characterisation of the Kerr
solution which is amenable to a 3+1 splitting.

\medskip
A number of invariant characterisations are known in the
literature. Each with their own advantages and disadvantages. For
completeness we discuss some which bear connection to the analysis
presented in this article:

\medskip
\noindent
\textbf{The Simon and Mars-Simon tensors.} A convenient way of
studying stationary solutions to the Einstein field equations is
through the quotient manifold of the orbits of the stationary Killing
vector. The Schwarzschild spacetime is characterised among all
stationary solutions by the vanishing of the Cotton tensor of the
metric of this quotient manifold ---see e.g. \cite{Fri04}. In
\cite{Sim84} a suitable generalisation of the Cotton tensor of the
quotient manifold was introduced ---the \emph{Simon tensor}. The vanishing of
the Simon tensor together with asymptotic flatness and non-vanishing
of the mass characterises the Kerr solution in the class of
stationary solutions. In \cite{Mar99,Mar00} a spacetime version of the
Simon tensor was introduced ---the so-called \emph{Mars-Simon tensor}. The
construction of this tensor requires the \emph{a priori} existence of
a Killing vector in the spacetime. Accordingly, it is tailored for the
problem of the uniqueness of stationary black holes. The vanishing of
the Mars-Simon tensor together with some global conditions (asymptotic
flatness, non-zero mass, stationarity of the Killing vector)
characterises the Kerr spacetime.

\medskip
\noindent
\textbf{Characterisations using concomitants of the Weyl tensor.} A
concomitant of the Weyl tensor is an object constructed from 
tensorial operations on the Weyl tensor and its covariant
derivatives. An invariant characterisation of the Kerr spacetime in
terms of concomitants of the Weyl tensor has been obtained in
\cite{FerSae09}. This result generalises a similar result for the
Schwarzschild spacetime given in \cite{FerSae98}. These
characterisations consist of a set of conditions on concomitants of
the Weyl tensor, which if satisfied, characterise locally the
Kerr/Schwarzschild spacetime. An interesting feature of the
characterisation is that it provides expressions for the stationary
and axial Killing vectors of the spacetime in terms of concomitants
of the Weyl tensor. Unfortunately, the concomitants used in the
characterisation are complicated, and thus, produce very involved
expressions when performing a 3+1 split.

\medskip
\noindent
\textbf{Characterisations by means of generalised symmetries.}
Generalised symmetries (sometimes also known as hidden symmetries) are
generalisations of the Killing vector equation ---like the Killing
tensors and conformal Killing-Yano tensors. These tensors arise
naturally in the discussion of the so-called Carter constant of motion
and in the separability of various types of linear equations on the
Kerr spacetime ---see e.g. \cite{Car68b,KamMcL84,PenRin86}. In
particular, the existence of a conformal Killing-Yano tensor is
equivalent to the existence of a valence-2 symmetric spinor satisfying
the Killing spinor equation. An important property of the
Schwarzschild and Kerr spacetimes is that they admit a Killing
spinor. This Killing spinor generates, in a certain sense, the
Killing vectors and Killing-Yano tensors of the exact solutions in question
\cite{HugSom73b}. Moreover, as it will be discussed in the main part of
this article, for a spacetime which is neither conformally flat nor of
Petrov type N, the existence of a Killing spinor associated to a
Killing-Yano tensor together with the requirement of asymptotic
flatness renders a characterisation of the Kerr spacetime. To the best
of our knowledge, this property has only been discussed in the
literature ---without proof--- in \cite{FerSae07}.

\medskip
Although at first sight independent, the characterisations of the
Schwarzschild and Kerr spacetimes described in the previous paragraphs
are interconnected ---sometimes in very subtle manners. This is not
too surprising as all these characterisations make use in a direct or
indirect manner of the fact that the Kerr spacetime is a vacuum
spacetime of Petrov type D ---see e.g. \cite{SKMHH} for a discussion
of the Petrov classification. The art in producing a useful
characterisation of the Kerr spacetime lies in finding further
conditions on type D spacetimes which are natural and simple to use.

\subsection*{A characterisation of Kerr data}

Characterisations of initial data sets for the Schwarzschild and Kerr
spacetimes have been discussed in
\cite{GarVal07,GarVal08b,Val05b}. These characterisations make use of
a number of local and global ingredients. For example, in 
\cite{GarVal08b} it is necessary to assume the existence of a Killing vector on
the development of the spacetime. 

\medskip
In this article we present a
rigorous and detailed discussion of a geometric invariant
characterising initial data for the Kerr spacetime. A restricted
version of this construction has been presented in \cite{BaeVal10a}.

\medskip
The starting point of our construction is the observation that the
existence of a Killing spinor in the Kerr spacetime is a key
property. It allows to relate the Killing vectors of the spacetime
with its curvature in a neat way. The reason for its importance can be
explained in the following way: from a specific Killing spinor it is
possible to obtain a Killing vector which in general will be
complex. It turns out that for the Kerr spacetime this Killing vector
is in fact real and coincides with the stationary Killing vector. It
can be shown that the Kerr solution is the only asymptotically flat
vacuum spacetime with these properties, if one assumes that there are
no points where the Petrov type is either N or O.

\medskip
Given the aforementioned spacetime characterisation of the Kerr
solution, the question now is how to make use of it to produce a
characterisation in terms of initial data sets. For this, one has to encode the
existence of a Killing spinor at the level of the data. The way of
doing this was first discussed in \cite{GarVal08a} and follows the
spirit of the well-known discussion of how to encode Killing vectors on
initial data ---see e.g. \cite{BeiChr97b}. 

\medskip
The conditions on the initial data that ensure the existence of a
Killing spinor in its development are called the \emph{Killing spinor
initial data equations} and are, like the Killing initial data
equations (KID equations), overdetermined. In \cite{Dai04c}, a
procedure was given on how to construct equations which generalise the
KID equations for time symmetric data. These generalised equations
have the property that for a particular behaviour at infinity they
always admit a solution. If the spacetime admits Killing vectors, then
the solutions to the generalised KID equations with the same
asymptotic behaviour as the Killing vectors are, in fact, Killing
vectors. Therefore, one calls the solutions to the generalised KID
equations \emph{approximate symmetries}. The total number of
approximate symmetries is equal to the maximal number of possible
Killing vectors on the spacetime. A peculiarity of this procedure is
that if the spacetime is not stationary, the approximate Killing
vector associated to a time translation does not have the same
asymptotic behaviour as a time translation\footnote{Here and in what
  follows, for a time translation it is understood a Killing vector
  which in some asymptotically Cartesian coordinate system has a
  leading term of the form $\partial_t$.}.

\medskip
The Killing spinor initial data equations consist of three
conditions: one of them differential (the \emph{spatial Killing spinor
equation})\footnote{The idea of using the spatial part of spinorial
equations to characterise slices of particular spacetimes is not
new. In \cite{Tod84} the spatial twistor equation has been used to
characterise slices of conformally flat spacetimes. See also
\cite{BeiSza97}.} and two \emph{algebraic conditions}. Following the
spirit of \cite{Dai04c} we construct a generalisation of the spatial
Killing spinor equation ---\emph{the approximate Killing spinor
equation}. This equation is elliptic and of second order. This
equation is the Euler-Lagrange equation of an integral functional
---the $L^2$-norm of the exact spatial Killing spinor equation. For
this equation it is possible to prove the following theorem:

\begin{main}
For initial data sets to the Einstein
field equations with suitable asymptotic behaviour, there exists a solution to the approximate Killing
spinor equation with the same asymptotic behaviour as the Killing
spinor of the Kerr spacetime.
\end{main}

A precise formulation will be given in the main text. In particular,
it will be seen that the conditions on the asymptotic behaviour of the
initial data are rather mild and amount to requiring the data to be,
in a sense, asymptotically Kerr data. Contrasted with the results in \cite{Dai04c}, this
result is notable because, arguably, the most important approximate
symmetry of \cite{Dai04c} does not share the same asymptotic behaviour
as the exact symmetry. The precise version of this theorem generalises
the one discussed in \cite{BaeVal10a} in that it allows for boosted
data. This generalisation is only possible after a detailed analysis
of the asymptotic solutions of the exact Killing spinor equation.

\medskip
The approximate Killing spinor discussed in the previous paragraphs
can be used to construct a geometric invariant for the initial
data. This invariant is global and  involves the $L^2$ norms
of the Killing spinor initial data equations evaluated at the
approximate Killing vector. It should be observed that only part of the invariant satisfies a
variational principle ---this is a further difference with respect to
the construction of \cite{Dai04c}. As the initial data set is assumed to be
asymptotically Euclidean,  one expects its development to be
asymptotically flat. This renders the desired characterisation of Kerr
data and our main result.

\begin{main}
Consider an initial data set for the vacuum Einstein field equations
whose development in a small slab is neither of Petrov type N nor O at
any point, and such that the $L^2$ norm of the Killing spinor initial
data equations evaluated at the solution (with the same asymptotic
behaviour as the Killing spinor of the Kerr spacetime) to the
approximate Killing spinor equation vanishes. Then the initial data
set is locally data for the Kerr spacetime. Furthermore, if the
3-manifold has the same topology as that of hypersurfaces of the Kerr
spacetime, then the initial data set is data for the Kerr spacetime.
\end{main}

There are several advantages of this characterisation over previous
ones given in the literature. Most notably, it allows to condense the
non-Kerrness of an initial data set in a single number. That this
invariant constitutes a good distance in the space of initial data
sets will be discussed elsewhere. Furthermore, the way the invariant
is constructed is fully amenable to a numerical implementation ---the
elliptic solvers that one would need to compute the solution to the
approximate Killing spinor equation are, nowadays, standard
technology.

\subsection*{Detailed outline of the article}
The outline of the article is as follows: in Section \ref{Section:Basics} we study
Killing spinors, and their influence on the algebraic type of the
spacetime. We relate the Killing spinors to Killing vectors and
Killing-Yano tensors. Using these results together with a
characterisation of the Kerr spacetime by Mars \cite{Mar00}, we conclude that the Kerr
spacetime can be characterised in terms of existence of a Killing
spinor related to a real Killing vector. This has previously been
overlooked in the literature, but it is a key element in our analysis.

Section \ref{Section:SpaceSpinors} follows with an exposition of space spinors, which will be
the main computational tool for the remainder of the paper.  Following
that, in Section \ref{Section:KSD} we study a 3+1 splitting of the Killing spinor
equation. A similar analysis was carried out in \cite{GarVal08a}, but
here we manage to condense the result into three simple equations, the
\emph{spatial Killing spinor equation} and two algebraic equations. We
also present general equations for the spatial derivatives of a
general valence 2 spinor, which is not necessarily a Killing
spinor. These equations are also used in later parts of the paper.

In Section \ref{Section:ApproximateKS} we introduce the new concept of \emph{approximate Killing
spinors}. These are introduced as solutions to an elliptic equation
formed by composing the spatial Killing spinor operator with its
formal adjoint. That this composed operator is indeed elliptic and
formally self adjoint is proved. We also see that the approximate
Killing spinor equation can be derived from a
variational principle.

To get unique solutions to the approximate Killing spinor equation, we
need to specify the asymptotic behaviour. For a rigorous treatment of
this, we use weighted Sobolev spaces; these are described in Section
\ref{Section:AsymptoticallyEuclideanData}.  Here we also study the asymptotics of a Killing spinor on a
boosted slice of the Schwarzschild spacetime. In general, we study
slices of an arbitrary spacetime with asymptotics similar to those of
the Schwarzschild spacetime. Using these assumptions, we can then in
Section \ref{Section:AB} prove existence of spinors with the same asymptotics as the
Killing spinor in the Schwarzschild spacetime. We later use these
spinors as seeds for solutions to the approximate Killing spinor
equation. In this way we get the desired asymptotic behaviour of our
approximate Killing spinors.

In Section \ref{Section:ApproximateKSinAEM} we study the approximate Killing spinor equation in our
asymptotically Euclidean manifolds to gain existence and uniqueness of
solutions with the desired asymptotics.  This is done by means of the
Fredholm alternative on weighted Sobolev spaces, transforming the
existence problem into a study of the kernel of the Killing spinor
operator.  In this process we get the first part of the geometric
invariant ---the $L^2$ norm of the approximate Killing spinor. This 
norm is proved to be finite. The geometric invariant is constructed in
Section \ref{Section:Invariant}, by adding the $L^2$ norms of the algebraic
conditions. There follows our main theorem: the  invariant
vanishes if and only if the spacetime is the Kerr spacetime. The
invariant is as a consequence of the construction proved to be finite
and well defined.

We also include two appendices. The first describes an alternative
proof of finiteness of a particular boundary integral in Section 
\ref{Section:ApproximateKSinAEM}. The other contains tensor versions of the invariant ---this can be
useful in applications.

\subsection*{General notation and conventions}
All throughout, $(\mathcal{M},g_{\mu\nu})$ will be an orientable and
time orientable globally hyperbolic vacuum spacetime. It follows that
the spacetime admits a spin structure ---see \cite{Ger68,Ger70c}. Here, and in
what follows, $\mu,\,\nu,\cdots$ denote abstract 4-dimensional tensor
indices. The metric $g_{\mu\nu}$ will be taken to have signature
$(+,-,-,-)$. Let $\nabla_\mu$ denote the Levi-Civita connection of
$g_{\mu\nu}$. The sign of the Riemann tensor will be given by the
equation
\[
\nabla_\mu\nabla_\nu\xi_\zeta-\nabla_\nu\nabla_\mu\xi_\zeta=R_{\nu\mu\zeta}{}^\eta\xi_\eta.
\]

\medskip
 The triple $(\mathcal{S}, h_{ab},K_{ab})$ will denote initial data on
a hypersurface of the spacetime $(\mathcal{M},g_{\mu\nu})$. The
symmetric tensors $h_{ab}$, $K_{ab}$ will denote, respectively, the
3-metric and the extrinsic curvature of the 3-manifold
$\mathcal{S}$. The metric $h_{ab}$ will be taken to be negative
definite ---that is, of signature $(-,-,-)$. The indices
$a,\,b,\ldots$ will denote abstract 3-dimensional tensor indices,
while $i,\,j,\ldots$ will denote 3-dimensional tensor coordinate
indices.  Let $D_a$ denote the Levi-Civita covariant derivative of
$h_{ab}$.

\medskip
Spinors will be used systematically. We follow the conventions of
\cite{PenRin84}.  In particular, $A,\,B,\ldots$ will denote abstract
spinorial indices, while $\mathbf{A}, \,\mathbf{B},\ldots$ will be
indices with respect to a specific frame. Tensors and their spinorial
counterparts are related by means of the solder form $\sigma_\mu{}^{AA'}$ satisfying $g_{\mu\nu}=\sigma_\mu^{AA'}\sigma_\nu^{BB'} \epsilon_{AB}\epsilon_{A'B'}$, where $\epsilon_{AB}$
is the antisymmetric spinor and $\epsilon_{A'B'}$ its complex
conjugate copy. One has, for example, that $\xi_\mu =
\sigma_{\mu}{}^{AA'} \xi_{AA'}$.  Let $\nabla_{AA'}$ denote the
spinorial counterpart of the spacetime connection
$\nabla_\mu$. Besides the connection $\nabla_{AA'}$, two other
spinorial connections will be used: $D_{AB}$, the spinorial
counterpart of the Levi-Civita connection $D_a$ and $\nabla_{AB}$,
the Sen connection of $(\mathcal{M},g_{\mu\nu})$ ---full details 
will be given in Section \ref{Section:SpaceSpinors}. 

\medskip
\noindent
\textbf{The Kerr spacetime.} For the Kerr spacetime it will be
understood the maximal analytic extension of the Kerr metric as
described by Boyer \& Lindquist \cite{BoyLin67} and Carter
\cite{Car68a}. When regarding the Kerr spacetime as the development of
Cauchy initial data, we will only consider its maximal globally
hyperbolic development.

\section{Killing spinors: general theory}
\label{Section:Basics}

As mentioned in the introduction, our point of departure will be a
characterisation of the Kerr spacetime based on the existence in the
spacetime of a valence-2 symmetric spinor satisfying the Killing
spinor equation. To the best of our knowledge, this characterisation of the Kerr spacetime has not explicitly been discussed in the literature, save for a side remark in \cite{FerSae07}. In this section we provide a summary of this characterisation and fill in some technical details. 

\subsection{Killing spinors and Petrov type D spacetimes}
A valence-2 Killing spinor is a symmetric spinor
$\kappa_{AB}=\kappa_{(AB)}$ satisfying the equation
\begin{equation}
\label{KillingSpinorEquation} \nabla_{A'(A} \kappa_{BC)}=0.
\end{equation} 
Killing spinors
offer a way of relating properties of the curvature to properties of
the symmetries of the spacetime. Taking a further derivative of
equation \eqref{KillingSpinorEquation}, antisymmetrising and
commuting the covariant derivatives one finds the integrability condition
\begin{equation}
\Psi_{(ABC}{}^F\kappa_{D)F}=0, \label{IntegrabilityCondition}
\end{equation}
where $\Psi_{ABCD}$ denotes the self-dual Weyl spinor. The above
integrability imposes strong restrictions on the algebraic type of
the Weyl spinor. More precisely, it follows that if $\Psi_{ABCD}\neq
0$ and $\kappa_{AB}\neq 0$, then 
\begin{equation}
\Psi_{ABCD} = \psi \kappa_{(AB} \kappa_{CD)}, \label{SolutionIntegrabilityCondition}
\end{equation}
where $\psi$ is a scalar. Thus, $\Psi_{ABCD}$ must be of Petrov type D
or N ---see e.g. \cite{GarVal08a,Jef84}. The converse is also true
\cite{HugPenSomWal72,PenRin86,WalPen70}.  Summarising:

\begin{theorem}[Walker \& Penrose 1970]
\label{Theorem:TypeDhasalwaysaKS}
A vacuum spacetime admits a valence-2 Killing spinor if and only if
it is of Petrov type D, N or O. 
\end{theorem}

 From \eqref{SolutionIntegrabilityCondition} it can also be seen that
$\Psi_{ABCD}$ is of Petrov type N if and only if $\kappa_{AB}$ is
  algebraically special. That is, there exists a spinor $\alpha_A$
  such that $\kappa_{AB}=\alpha_A \alpha_B$. Thus, an algebraically
  general Killing spinor $\kappa_{AB}=\alpha_{(A}\beta_{B)}$ is
  always associated to a vacuum spacetime of Petrov type D.

\subsection{The Killing vector associated to a Killing spinor and the generalised Kerr-NUT metrics}

Given a Killing spinor $\kappa_{AB}$, the concomitant 
\begin{equation}
\label{ComplexKillingVector}
\xi_{AA'}=\nabla^B{}_{A'} \kappa{}_{AB},
\end{equation}
 is a complex Killing vector of the spacetime: its real and imaginary
parts are themselves Killing vectors of the spacetime
\cite{HugSom73b}.  In relation to this it should be pointed out that
all vacuum Petrov type D spacetimes are known \cite{Kin69}. It follows
from the analysis in the latter reference that all vacuum, Petrov type
D spacetimes have a commuting pair of Killing vectors.  A key property
of the Kerr spacetime is the following
(cfr. \cite{HugSom73b,PenRin86}):

\begin{proposition}
\label{Proposition:KSrendersKV}
Let $(\mathcal{M},g_{\mu\nu})$ be a vacuum Petrov type D spacetime. The
Killing vector $\xi_{AA'}$ given by \eqref{ComplexKillingVector} is
real in the case of the Kerr spacetime.
\end{proposition}

\medskip
\noindent
\textbf{Remark 1.} In what follows, the class of Petrov type D
spacetimes for which $\xi_{AA'}$ is real will be called the
\emph{generalised Kerr-NUT class} ---cfr. \cite{FerSae07}. This class
can be alternatively characterised ---see e.g. \cite{KamMcL84}--- by the
existence of a Killing-Yano tensor
\[
Y_{\mu\nu}=Y_{[\mu\nu]}, \quad \nabla_{(\mu} Y_{\nu)\lambda}=0.
\]
The correspondence between the Killing spinor $\kappa_{AB}$ and the
spinorial counterpart $Y_{AA'BB'}$ of the Killing-Yano tensor, $Y_{\mu\nu}$, is given by
\[
Y_{AA'BB'} \equiv\mbox{i} \left( \kappa_{AB}\epsilon_{A'B'} - \epsilon_{AB} \bar{\kappa}_{A'B'}  \right),
\] 
where the overbar denotes the complex conjugate.

\medskip
\noindent
\textbf{Remark 2.} In terms of the Kinnersley list of type D metrics,
the class of generalised Kerr-NUT metrics contains, in addition to the proper
Kerr-NUT metrics (II.C), also the metrics II.E ---see \cite{DebKamMcL84}.

\medskip
An important property of the generalised Kerr-NUT metrics involves the
Killing form,\linebreak $F_{AA'BB'}=-F_{BB'AA'}$, of a real Killing vector
$\xi_{AA'}$ defined by
\begin{equation}
\label{KillingForm}
F_{AA'BB'} \equiv \frac{1}{2}\left( \nabla_{AA'}\xi_{BB'} - \nabla_{BB'}\xi_{AA'} \right).
\end{equation}
Let
\begin{equation}
\label{SelfDualKillingForm}
\mathcal{F}_{AA'BB'} \equiv \frac{1}{2}\left(F_{AA'BB'} + \mbox{i}F^*_{AA'BB'}   \right)
\end{equation}
denote the corresponding \emph{self-dual Killing form}, with
$F^*_{AA'BB'}$ the Hodge dual of $F_{AA'BB'}$. Due to the symmetries of the Killing form one can write
\begin{equation}
\mathcal{F}_{AA'BB'} = \mathcal{F}_{AB} \epsilon_{A'B'},
\end{equation}
with 
\begin{equation}
\label{KFSpinor}
\mathcal{F}_{AB} \equiv \frac{1}{2} F_{AQ'B}{}^{Q'} = \mathcal{F}_{BA}. 
\end{equation}
One has the following result

\begin{lemma}
\label{Lemma:PrincipalDirections}
For generalised Kerr-NUT spacetimes one has that
\[
\mathcal{F}_{AB} = \varkappa \kappa_{AB},
\]
where $\varkappa$ is a non-vanishing scalar function, so that the
principal spinors of $\mathcal{F}_{AB}$ and $\Psi_{ABCD}$ are
parallel. Equivalently, one has that 
\[
\Psi_{ABPQ} \mathcal{F}^{PQ} = \varphi \mathcal{F}_{AB},
\]
with $\varphi$ a non-vanishing scalar. 
\end{lemma}

\begin{proof}
 One proceeds by a direct computation. One notes
that the expressions \eqref{KillingForm},
\eqref{SelfDualKillingForm} and \eqref{KFSpinor} assume that the
Killing vector $\xi_{AA'}$ is real. Using equations \eqref{ComplexKillingVector} and \eqref{KFSpinor} and the vacuum commutators for $\nabla_{AA'}$ one finds that
\[
\mathcal{F}_{AB} = \frac{3}{4} \Psi_{ABPQ} \kappa^{PQ}. 
\]
As the spacetime is assumed to be of Petrov type D one has that
$\kappa_{AB}=\alpha_{(A}\beta_{B)}$ with $\alpha_A \beta^A
=\varsigma$, where $\varsigma$ is a non-vanishing scalar. From equation \eqref{SolutionIntegrabilityCondition} one finds then that
$\Psi_{ABCD}\negthinspace =\psi \alpha_{(A}\alpha_B \beta_C \beta_{D)}$, so that
\[
\Psi_{ABPQ} \kappa^{PQ} = - \frac{1}{3} \psi\varsigma^2 \kappa_{AB}, 
\] 
and finally that
\[
\mathcal{F}_{AB} = -\frac{1}{4} \psi \varsigma^2 \kappa_{AB},
\]
 from where the desired result follows.
\end{proof}

\medskip
The property that allows us to single out the Kerr spacetime out of
the generalised Kerr-NUT class is given by the following result proved by Mars
\cite{Mar99,Mar00}.

\begin{theorem}[Mars 1999, 2000]
\label{Theorem:MMars}
Let $(\mathcal{M},g_{\mu\nu})$ be a smooth vacuum spacetime with the following properties:
\begin{itemize}
\item[(i)] $(\mathcal{M},g_{\mu\nu})$ admits a Killing vector $\xi_{AA'}$ such that, $\mathcal{F}_{AB}$, the spinorial counterpart of the Killing form of $\xi_{AA'}$ satisfies
\[
\Psi_{ABPQ} \mathcal{F}^{PQ} = \varphi \mathcal{F}_{AB},
\]
with $\varphi$ a scalar;

\item[(ii)] $(\mathcal{M},g_{\mu\nu})$ contains a stationary
asymptotically flat 4-end, and $\xi_{AA'}$ tends to a time translation at
infinity, and the Komar mass of the asymptotic end is non-zero.
 \end{itemize}
Then $(\mathcal{M},g_{\mu\nu})$ is locally isometric to the Kerr spacetime. 
\end{theorem}

\medskip
\noindent
\textbf{Remark.} A stationary asymptotically flat 4-end is an open
submanifold $\mathcal{M}_\infty \subset \mathcal{M}$ diffeomorphic to
$I\times (\Real^3 \setminus \mathcal{B}_R)$, where $I\subset \Real$ is
an open interval and $\mathcal{B}_R$ a closed ball of radius $R$ such
that in the local coordinates $(t,x^i)$ defined by the diffeomorphism, the
metric $g_{\mu\nu}$ satisfies 
\begin{eqnarray*}
&& |g_{\mu\nu}-\eta_{\mu\nu}| + |r\partial_i g_{\mu\nu}| \leq C
r^{-\alpha}, \\
&& \partial_t g_{\mu\nu} =0,
\end{eqnarray*} 
with $C$, $\alpha$ constants, $\eta_{\mu\nu}$ is the Minkowski
metric and $r=\sqrt{(x^1)^2 + (x^2)^2 + (x^3)^2}$. In particular
$\alpha\geq 1$. The definition of the Komar mass is given in
\cite{Kom58}. In this context it coincides with the ADM mass of the
spacetime.

\subsection{Non-degeneracy of the Petrov type of the Kerr spacetime}
Finally, we note the following result about the non-degeneracy of the
Petrov type of the Kerr spacetime \cite{Mar00}.

\begin{proposition}[Mars 2000]
\label{Proposition:Kerrdoesnotdegenerate}
The Petrov type of the Kerr spacetime is always D ---there are no
points where it degenerates to type N or O. 
\end{proposition}

\subsection{A characterisation of the Kerr spacetime using Killing spinors}

As a consequence of Theorem \ref{Theorem:TypeDhasalwaysaKS} and propositions \ref{Proposition:KSrendersKV}, \ref{Proposition:Kerrdoesnotdegenerate}
one obtains the following invariant characterisation of the Kerr
spacetimes. From this characterisation we will extract, in the sequel, a
characterisation of asymptotically Euclidean Kerr data.

\begin{theorem}
\label{Theorem:SpacetimeCharacterisation}
 Let $(\mathcal{M},g_{\mu\nu})$ be a smooth vacuum 
 spacetime such that 
\begin{equation*}
\Psi_{ABCD}\neq 0 ,\qquad
\Psi_{ABCD}\Psi^{ABCD}\neq 0 
\end{equation*}
on $\mathcal{M}$.
Then $(\mathcal{M},g_{\mu\nu})$ is locally 
isometric to the Kerr spacetime if and only if the following
conditions are satisfied:
\begin{itemize}
\item[(i)] there exists a Killing spinor, $\kappa_{AB}$, such that the associated Killing
  vector, $\xi_{AA'}$, is real;
\item[(ii)] the spacetime $(\mathcal{M},g_{\mu\nu})$ has a
  stationary asymptotically flat 4-end with non-vanishing mass in which $\xi_{AA'}$
  tends to a time translation.
\end{itemize}
\end{theorem}

\begin{proof}
Clearly, the conditions (i) and (ii) are necessary to obtain the Kerr
spacetime. For the sufficiency, assume that (i) holds, that is, the spacetime has a
Killing spinor $\kappa_{AB}$ such that the associated Killing vector
$\xi_{AA'}$ is real. Accordingly, the spacetime must be of type $D$, $N$ or
$O$. As $\Psi_{ABCD}\neq 0$ and
$\Psi_{ABCD}\Psi^{ABCD}\neq 0$ by hypothesis, the spacetime cannot be
of types $N$ or $O$. By the reality of $\xi_{AA'}$ it must be a
generalised Kerr-NUT spacetime and the conclusion of Lemma
\ref{Lemma:PrincipalDirections} follows. Now, if (ii) holds then by
Theorem \ref{Theorem:MMars}, the spacetime has to be locally the Kerr spacetime. 

\end{proof}

\medskip
\noindent
\textbf{Remark.} It is of interest to see whether the conditions
$\Psi_{ABCD}\neq 0$ and $\Psi_{ABCD}\Psi^{ABCD}\neq 0$ can be
removed. An analysis along what is done in the proof of Theorem
\ref{Theorem:MMars} may allow to do this. This will be discussed
elsewhere.

\section{Space spinors: general theory}
\label{Section:SpaceSpinors}

As mentioned in the introduction, in this article we will make use of
a space spinor formalism to project the longitudinal and transversal
parts of the Killing spinor equation \eqref{KillingSpinorEquation}
with respect to the timelike vector field $\tau^\mu$. The space spinor
formalism was originally introduced in \cite{Som80}. Here we follow conventions and notations similar to those in \cite{GarVal08a}. For
completeness, we introduce all the relevant notation here.

\subsection{Basic definitions}
Let $\tau^\mu$ be a timelike vector field on $(\mathcal{M},g_{\mu\nu})$
with normalisation $\tau_\mu \tau^\mu=2$. Define the projector
\[
h_{\mu\nu}\equiv g_{\mu\nu} -\frac{1}{2} \tau_\mu\tau_\nu.  
\]
We also define the following tensors:
\begin{eqnarray*}
&& K_{\mu\nu} = -h_\mu{}^\lambda h_\nu{}^\rho \nabla_\lambda
\tau_\rho, \\
&& K^\mu = -\frac{1}{2} \tau^\nu \nabla_\nu \tau^\mu.
\end{eqnarray*}
Note that it is not being assumed that $\tau^\mu$ is hypersurface
orthogonal. Thus, the tensor $K_{\mu\nu}$ as defined above is not
necessarily the second fundamental form of a foliation of the
spacetime $(\mathcal{M},g_{\mu\nu})$. 

\medskip
Let $\tau^{AA'}$ denote the
spinorial counterpart of $\tau^\mu$. One has that $\tau^{AA'}\equiv
\sigma_{\mu}{}^{AA'}\tau^\mu$ so that
\[
 \tau_{AA'}\tau^{AA'}=2, \quad \tau^A{}_{A'}\tau^{BA'}=\epsilon^{AB}.
\]
 The spinor $\tau^{AA'}$ allows to introduce the \emph{spatial solder
forms}   
\[
\sigma_\mu{}^{AB}\equiv\sigma_\mu{}^{(A}{}_{A'}\tau^{B)A'}, \quad 
\sigma^\mu{}_{AB} \equiv \tau_{(B}{}^{A'} \sigma^\mu{}_{A)A'},
\]
so that one has 
\begin{eqnarray*}
&& \sigma^\mu{}_{AB} \sigma_\nu{}^{AB}=h^\mu{}_\nu, \quad g_{\mu\nu}\sigma^\mu{}_{AB} \sigma^\nu{}_{CD}=h_{\mu\nu}\sigma^\mu{}_{AB} \sigma^\nu{}_{CD} = \frac{1}{2}( \epsilon_{AC}\epsilon_{BD} + \epsilon_{AD}\epsilon_{BC}), \\
&& \quad \tau_\mu\sigma^\mu{}_{AB}=0, \quad \epsilon_{AB}\epsilon_{A'B'} = \frac{1}{2}\tau_{AA'}\tau_{BB'} + h_{\mu\nu} \sigma^\mu{}_{AE} \sigma^\nu{}_{BF} \tau^E{}_{A'}\tau^F{}_{B'}.
\end{eqnarray*}

\medskip
If $\tau^\mu$ is hypersurface orthogonal, then $h_{ab}$, $K_{ab}$,
$K^a$, $\sigma_a{}^{\mathbf{AB}}$, $\sigma^a{}_{\mathbf{AB}}$ denote, respectively , the
pull-backs to the hypersurfaces orthogonal to $\tau^\mu$ of $h_{\mu\nu}$,
$K_{\mu\nu}$, $K^\mu$, $\sigma_\mu{}^{\mathbf{AB}}$,
$\sigma^\mu{}_{\mathbf{AB}}$ ---note that these objects are spatial, in
the sense that their contraction with $\tau^\mu$ vanishes, and thus,
their pull-backs are well defined.  The
relevant properties of these tensors apply to their pull-backs. Often
we will begin with a spacelike hypersurface $\mathcal{S}$, and define
$\tau^\mu$ as the normal to this hypersurface, we then automatically get the desired properties.

\subsection{Space spinor splittings}
The spinor $\tau^{AA'}$ can be used to construct a formalism
consisting of unprimed indices. For example, given a spacetime spinor
$\zeta_{AA'}$ one can write
\begin{equation}
\label{SpaceSpinorSplit}
\zeta_{AA'} = \frac{1}{2}\tau_{AA'} \zeta-
\tau_{A'}{}^P \zeta_{PA},
\end{equation}
with
\[
\zeta \equiv \tau^{PP'}\zeta_{PP'}, \quad \zeta_{AB} \equiv \tau_{(A}{}^{P'}\zeta_{B)P'}.
\]
This decomposition can be extended in a direct manner to higher
valence spinors.  Any spatial tensor has a space-spinor
counterpart. For example, if $T_\mu{}^\nu$ is a spatial tensor
(i.e. $\tau^\mu T_\mu{}^\nu=0$ and $\tau_\nu T_\mu{}^\nu=0$), then
its  space spinor counterpart is given by $T_{AB}{}^{CD}=\sigma^\mu{}_{AB}\sigma_\nu{}^{CD}T_\mu{}^\nu$. 

\subsection{Spinorial covariant derivatives}
Applying formally the space spinor split given by
\eqref{SpaceSpinorSplit} to the spacetime spinorial covariant
derivative $\nabla_{AA'}$ one obtains
\[
\nabla_{AA'}=\frac{1}{2}\tau_{AA'}\nabla-\tau_{A'}{}^B\nabla_{AB} ,
\]
where we have introduced the differential operators 
\begin{eqnarray*}
&& \nabla \equiv \tau^{AA'}\nabla_{AA'},\\
&& \nabla_{AB} \equiv \tau^{A'}{}_{(A}\nabla_{B)A'}=\sigma^\mu{}_{AB}\nabla_\mu.
\end{eqnarray*}
The latter is referred to as the \emph{Sen connection}. Let $K_{ABCD}$
denote the space spinor counterpart of the tensor $K_{\mu\nu}$.  One
has that 
\[
K_{ABCD}=\tau_D{}^{C'}\nabla_{AB}\tau_{CC'}, \quad K_{ABCD}=K_{(AB)(CD)}.
\]
In the sequel, it will be convenient to write $K_{ABCD}$ in terms of
its irreducible components. For this define
\[
\Omega_{ABCD}\equiv K_{(ABCD)}, \quad  \Omega_{AB}\equiv
K_{(A}{}^C{}_{B)C},  \quad K\equiv\updn{K}{AB}{AB},
\] 
so that one can write
\begin{equation}
\label{KSplit}
K_{ABCD}=\Omega_{ABCD}-\frac{1}{2}\epsilon_{A(C}\Omega_{D)B}-\frac{1}{2}\epsilon_{B(C}\Omega_{D)A}-\frac{1}{3}\epsilon_{A(C}\epsilon_{D)B}K,
\end{equation}
If $\tau^\mu$ is hypersurface orthogonal, then $\Omega_{AB}=0$, and
thus $K_{\mu\nu}$ can be regarded as the extrinsic curvature of the
leaves of a foliation of the spacetime $(\mathcal{M},g_{\mu\nu})$. Let
$K_{AB}$ denote the spinorial counterpart of the acceleration $K_\mu$. It has the symmetry $K_{AB}=K_{(AB)}$ 
and satisfies  
\[
K_{AB}=\tau_B{}^{A'}\nabla\tau_{AA'} .
\]

\medskip
If $\tau^\mu$ is hypersurface orthogonal then the pull-back, $D_a$, of  $D_\mu \equiv
h^\nu{}_\mu \nabla_\nu$ corresponds to the Levi-Civita connection of
the intrinsic metric of the leaves of the foliation of hypersurfaces
orthogonal to $\tau^\mu$. Its spinorial counterpart is given by 
$D_{AB}=D_{(AB)}=\sigma^a{}_{AB}D_a$. The Sen connection,
$\nabla_{AB}$, and the Levi-Civita connection, $D_{AB}$, are related to
each other through the spinor $K_{ABCD}$. For example, for a valence 1
spinor $\pi_C$ one has that
\[
\nabla_{AB} \pi_C = D_{AB}\pi_C + \frac{1}{2} K_{ABC}{}^{D} \pi_D,
\]
with the obvious generalisations for higher valence spinors.

\subsection{Hermitian conjugation}
Given a spinor $\pi_{A}$, we define its \emph{Hermitian conjugate} via
\[
\hat{\pi}_{A} \equiv \tau_{A}{}^{E'}\bar{\pi}_{E'}.
\] 
The Hermitian conjugate can be extended to higher valence symmetric spinors in the
obvious way.  The spinors $\nu_{AB}$ and $\xi_{ABCD}$ are said to be
real if 
\[
\hat{\nu}_{AB}=-\nu_{AB},\quad \hat{\xi}_{ABCD}=\xi_{ABCD}.
\]
It can be verified that $\nu_{AB}
\hat{\nu}^{AB}, \; \xi_{ABCD} \hat{\xi}^{ABCD}\geq 0$. If the spinors
are real, then there exist real spatial tensors $\nu_a$, $\xi_{ab}$ such that
$\nu_{AB}$ and $\xi_{ABCD}$ are their spinorial counterparts. 

Notice that the differential operator $D_{AB}$ is real in the sense
that
\[
\widehat{D_{AB}\pi_C}=-D_{AB}\hat\pi_C.
\]
Crucially, however, one has that
\[
\widehat{\nabla_{AB}\pi_C}=-\nabla_{AB}\hat\pi_C + \tfrac{1}{2}K_{ABC}{}^D\hat\pi_D.
\]

\subsection{Commutators}
The analysis in the sequel will require intensive use of the
commutators of the covariant derivative operators $\nabla$ and
$\nabla_{AB}$. These can be derived from a space spinor splitting of
the commutator of $\nabla_{AA'}$. 

\medskip
Define
\[
\square_{AB} \equiv \nabla_{C'(A}\nabla_{B)}{}^{C'}, \quad
\widehat{\square}_{AB} \equiv \tau_A{}^{A'} \tau_B{}^{B'}\square_{A'B'}=\tau_A{}^{A'} \tau_B{}^{B'}\nabla_{C(A'}\nabla_{B')}{}^{C}.
\]
The action of these operators on a spinor $\pi_A$ is given by
\[
\square_{AB}\pi_C =\Psi_{ABCQ}\pi^Q + \tfrac{1}{2}\Lambda
\epsilon_{C(A}\pi_{B)}, \quad \widehat{\square}_{AB} \pi_C
=\tau_A{}^{A'}\tau_B{}^{B'} \Phi_{FCA'B'}\pi^F,
\]
where $\Phi_{ABA'B'}$ and $\Lambda$ denote respectively, the spinor
counterparts of the tracefree part of the Ricci tensor $R_{\mu\nu}$ and the Ricci scalar $R$
of the spacetime metric $g_{\mu\nu}$. Clearly, the above expressions
simplify in the case of a vacuum spacetime, where we have
$\Phi_{ABA'B'}=0$, $\Lambda=0$.

\medskip
In terms of $\square_{AB}$ and $\widehat{\square}_{AB}$,  the
commutators of $\nabla$ and $\nabla_{AB}$ read
\begin{subequations}
\begin{align}
[\nabla,\nabla_{AB}] ={}&
\widehat\square_{AB}-\square_{AB}-\tfrac{1}{2}K_{AB}\nabla+K^D{}_{(A}\nabla_{B)D}-K_{ABCD}\nabla^{CD},\label{commutator1}\\
[\nabla_{AB},\nabla_{CD}] ={}& \frac{1}{2}\left( 
\epsilon_{A(C}\square_{D)B} + \epsilon_{B(C}\square_{D)A} 
\right) +\frac{1}{2}\left( 
\epsilon_{A(C}\widehat\square_{D)B} 
+\epsilon_{B(C}\widehat\square_{D)A} 
\right)\nonumber \\
&+\frac{1}{2}(K_{CDAB}\nabla-K_{ABCD}\nabla) 
+K_{CDQ(A}\nabla_{B)}{}^Q-K_{ABQ(C}\nabla_{D)}{}^Q. \label{commutator2}
\end{align}
\end{subequations}

\subsection{Decomposition of the Weyl spinor}
The Hermitian conjugation can  be used to decompose the Weyl spinor
$\Psi$ in terms of its electric and magnetic parts via
\[
E_{ABCD} \equiv \frac{1}{2}\left( \Psi_{ABCD} +
  \hat{\Psi}_{ABCD}\right), \quad B_{ABCD}\equiv
\frac{\mbox{i}}{2}\left(\hat{\Psi}_{ABCD} - \Psi_{ABCD} \right),
\]
so that
\[
\Psi_{ABCD} = E_{ABCD} + \mbox{i}B_{ABCD}. 
\]
The spinorial Bianchi identity  $\nabla^{AA'}\Psi_{ABCD}=0$ can be
split using the space spinor formalism to render
\begin{subequations}
\begin{eqnarray}
&& \nabla\Psi_{ABCD}=2\nabla^E{}_A\Psi_{BCDE}, \label{Bianchi1}\\
&& \nabla^{AB}\Psi_{ABCD}=0. \label{Bianchi2}
\end{eqnarray}
\end{subequations}

\medskip
Crucial for our applications is that the spinors $E_{ABCD}$ and
$B_{ABCD}$ can be expressed in terms of quantities intrinsic to a
hypersurface $\mathcal{S}$. More precisely, if $\Omega_{AB}=0$, one has that 
\begin{subequations}
\begin{eqnarray}
&&  E_{ABCD}= -r_{(ABCD)} + \tfrac{1}{2}\Omega_{(AB}{}^{PQ}\Omega_{CD)PQ}
- \tfrac{1}{6}\Omega_{ABCD}K, \label{Weyl:Electric}\\
&&  B_{ABCD}=-\mbox{i}\ D^Q{}_{(A}\Omega_{BCD)Q}, \label{Weyl:Magnetic}
\end{eqnarray}
\end{subequations}
where $r_{ABCD}$ is the space spinor counterpart of the  Ricci tensor
of the intrinsic metric of the hypersurface $\mathcal{S}$. 

\subsection{Space spinor expressions in Cartesian coordinates}
In some occasions it will be necessary to give spinorial expressions
in terms of Cartesian or asymptotically Cartesian frames and
coordinates. For this we make use of the spatial Infeld-van der
Waerden symbols $\sigma^{i}{}_{\mathbf{A}\mathbf{B}}$,
$\sigma_{i}{}^{\mathbf{A}\mathbf{B}}$. Given $x^{i}, \; \xi_{i}\in
\Real^3$ we shall follow the convention that
\[
x^{\mathbf{AB}} \equiv
\sigma_{i}{}^{\mathbf{A}\mathbf{B}} x^{i},
\quad \xi_{\mathbf{AB}} \equiv
\sigma^{i}{}_{\mathbf{A}\mathbf{B}} \xi_{i},
\]
with
\begin{equation}
x^{\mathbf{AB}}= \frac{1}{\sqrt{2}}
\left(
\begin{array}{cc}
-x^1 + \mbox{i}x^2 & x^3 \\
x^3 & x^1 +\mbox{i} x^2
\end{array}
\right),
\quad 
\xi_{\mathbf{AB}}= \frac{1}{\sqrt{2}}
\left(
\begin{array}{cc}
-\xi_1 - \mbox{i}\xi_2 & \xi_3 \\
\xi_3 & \xi_1 -\mbox{i}\xi_2
\end{array}
\right).
\label{CartesianSpinor}
\end{equation}

\section{Killing spinor data}
\label{Section:KSD}
In this section we review some aspects of the space spinor
decomposition of the Killing spinor equation \eqref{KillingSpinorEquation}. A first
analysis along these lines was first carried out in \cite{GarVal08a}. The
current presentation is geared towards the construction of geometric
invariants.

\subsection{General observations}
Given a symmetric spinor $\kappa_{AB}$ (not necessarily a Killing
spinor), it will be convenient to define the following spinors:
\begin{subequations}
\begin{align}
\xi &\equiv \nabla^{PQ}\kappa_{PQ},\label{xi_sen_1} \\
\xi_{BF} &\equiv \frac{3}{2}\nabla_{(F}{}^{D}\kappa_{B)D},\label{xi_sen_2}\\
\xi_{ABCD} &\equiv \nabla_{(AB}\kappa_{CD)}\label{xi_sen_3},\\
\xi_{AA'} &\equiv \nabla^B{}_{A'}\kappa_{AB},\\
H_{A'ABC} &\equiv 3 \nabla_{A'(A}\kappa_{BC)},\\
S_{AA'BB'} &\equiv \nabla_{AA'}\xi_{BB'} + \nabla_{BB'}\xi_{AA'}.
\end{align}
\end{subequations}
We will use this notation throughout the rest of the paper.
Clearly, for a Killing spinor one has
\[
H_{A'ABC}=0, \quad S_{AA'BB'}=0. 
\]
The spinors $\xi$, $\xi_{AB}$ and $\xi_{ABCD}$ arise in the
space spinor decomposition of the spinors $H_{A'ABC}$ and
$\xi_{AA'}$. To see this, let $\tau^{AA'}$ denote,  as in section \ref{Section:SpaceSpinors}, 
the spinorial counterpart of a timelike vector with normalisation
$\tau_{AA'}\tau^{AA'}=2$. Some manipulations show that
\begin{subequations}
\begin{eqnarray}
&&
\xi_{AA'}=\tfrac{1}{2}\tau_{AA'}\xi-\tfrac{2}{3}\tau^B{}_{A'}\xi_{AB}+\tfrac{1}{2}\tau^B{}_{A'}\nabla\kappa_{AB},
\label{Split:xi}\\
&&
H_{A'ABC}=\tau_{A'(A}\xi_{BC)}+\tfrac{3}{2}\tau_{A'(A}\nabla\kappa_{BC)}-3\tau_{A'}{}^D\xi_{ABCD}. \label{Split:H}
\end{eqnarray}
\end{subequations}
Furthermore, the spinors $\xi$, $\xi_{AB}$ and $\xi_{ABCD}$ correspond
to the irreducible components of $\nabla_{AB}\kappa_{CD}$ so that one
can write:
\begin{equation}\label{SenDiffKappaSplit}
\nabla_{AB}\kappa_{CD}=\xi_{ABCD}-\tfrac{1}{3}\epsilon_{A(C}\xi_{D)B}-\tfrac{1}{3}\epsilon_{B(C}\xi_{D)A}-\tfrac{1}{3}\epsilon_{A(C}\epsilon_{D)B}\xi.
\end{equation}

Using the commutator \eqref{commutator2} for vacuum one can obtain
equations for the derivatives of $\xi$ and $\xi_{AB}$ ---these will be
used systematically in the sequel. The irreducible components of the
derivative $\nabla_{AB}\xi_{CD}$ are given by:
\begin{subequations}
\begin{align}
\nabla^{AB}\xi_{AB}={}&
-\tfrac{1}{2}K\xi+\tfrac{3}{4}\Omega^{ABCD}\xi_{ABCD}
+\tfrac{1}{2}\Omega^{AB}\xi_{AB}
-\tfrac{3}{4}\Omega^{AB}\nabla\kappa_{AB}, \label{Dxi1}\\
\nabla^C{}_{(A}\xi_{B)C}={}&
\nabla_{AB}\xi
+\tfrac{3}{2}\Psi_{ABCD}\kappa^{CD} 
-\tfrac{2}{3}K\xi_{AB} 
-\tfrac{1}{2}\Omega_{ABCD}\xi^{CD} 
-\tfrac{3}{2}\xi_{(A}{}^{CDF}\Omega_{B)CDF}\nonumber\\
&-\tfrac{3}{2}\nabla^{CD}\xi_{ABCD} 
-\tfrac{1}{2}\Omega_{AB}\xi
+\tfrac{1}{2}\Omega_{(A}{}^C\xi_{B)C}
+\tfrac{3}{4}\Omega^{CD}\xi_{ABCD}
-\tfrac{3}{2}\Omega_{(A}{}^C\nabla\kappa_{B)C},
\label{Dxi2}\\ 
\nabla_{(AB}\xi_{CD)}={}& 
3\Psi_{F(ABC}\kappa_{D)}{}^F 
+K\xi_{ABCD}
-\tfrac{1}{2}\Omega_{ABCD}\xi 
+\Omega_{(ABC}{}^F\xi_{D)F} 
-\tfrac{3}{2}\Omega^{PQ}{}_{(AB}\xi_{CD)PQ} \nonumber \\
&+3\nabla^Q{}_{(A}\xi_{BCD)Q}
+\tfrac{1}{2}\Omega_{(AB}\xi_{CD)}
-\tfrac{3}{2}\Omega^F{}_{(A}\xi_{BCD)F}
+\tfrac{3}{2}\Omega_{(AB}\nabla\kappa_{CD)}. 
\label{Dxi3}
\end{align}
\end{subequations}
We note the appearance of the term $\nabla_{AB}\xi$ in
\eqref{Dxi2}. Thus, there is no independent equation for the
derivative of $\xi$. 

Finally, we consider the equations for the second order derivatives of
$\xi$. For the sake of simplicity, we restrict our attention to the
case when $\Omega_{AB}=0$ so that $K_{ABCD}=K_{CDAB}$.  For notational
purposes we define
$\Omega_{ABCDEF}\equiv\nabla_{(AB}\Omega_{CDEF)}$. One finds:
\begin{subequations}
\begin{align}
\nabla^{AB} \nabla_{AB} \xi ={}& 
-\tfrac{1}{6}K^2\xi
-\tfrac{1}{2}\Omega^{ABCD}\Omega_{ABCD}\xi
+3\Psi_A{}^{CDF}\Omega_{BCDF}\kappa^{AB}
+\xi_{AB}\nabla^{AB}K\nonumber\\
&+\tfrac{3}{4}\hat\Psi^{ABCD}\xi_{ABCD}
-\tfrac{9}{4}\Psi^{ABCD}\xi_{ABCD}
+2 K \Omega^{ABCD}\xi_{ABCD}\nonumber\\
&-\tfrac{15}{4}\Omega^{ABFH}\Omega^{CD}{}_{FH}\xi_{ABCD}
+\tfrac{9}{2}\Omega^{ABCD}\nabla^F{}_D\xi_{ABCF} \nonumber \\
&+\tfrac{3}{2}\nabla^{AB}\nabla^{CD}\xi_{ABCD},\label{nabla2xi0a}\\
\nabla^C{}_{(A}\nabla_{B)C}\xi ={}
&\tfrac{1}{2}\Omega_{ABCD}\nabla^{CD}\xi
-\tfrac{1}{3}K \nabla_{AB}\xi, \label{nabla2xi0b}\\
\nabla_{(AB}\nabla_{CD)}\xi={}&
-4 K \Psi_{(ABC}{}^E \kappa_{D)E} 
+\tfrac{1}{2} \hat\Psi_{ABCD}\xi 
-\tfrac{5}{2} \Psi_{ABCD}\xi 
-\tfrac{2}{3} \hat\Psi_{(ABC}{}^E \xi_{D)E} \nonumber\\ 
& -\tfrac{10}{3} \Psi_{(ABC}{}^E \xi_{D)E}
+\Omega_{ABCDEL} \xi^{EL} 
+\tfrac{4}{3} K{}^2 \xi_{ABCD}
+3 \Omega_{EFL(ABC}\xi_{D)}{}^{ELF}\nonumber\\
&+3 \Psi_{(AB}{}^{EL} \xi_{CD)EL} 
-\tfrac{3}{2}\xi_{(A}{}^{ELF}\Omega_{BCD)}{}^H\Omega_{ELFH} 
-3 \Psi_{EL(A}{}^F \kappa ^{EL} \Omega_{BCD)F}\nonumber\\
&-\xi^{EL} \Omega_{ELF(A} \Omega_{BCD)}{}^F 
+\tfrac{2}{3}K\xi_{(A}{}^E\Omega_{BCD)E}
+\tfrac{1}{2} \xi ^{ELFH} \Omega_{EL(AB} \Omega_{CD)FH}\nonumber\\ 
&-3 \Psi_{E(B}{}^{LF} \kappa_A{}^E \Omega_{CD)LF} 
-3 \Psi_{E(AB}{}^F \kappa^{EL}\Omega_{CD)LF}
-\Omega_{ELF(B}\xi_A{}^E \Omega_{CD)}{}^{LF}\nonumber\\ 
&-4 K \xi_{(AB}{}^{EL} \Omega_{CD)EL} 
-\tfrac{1}{2}\xi\Omega_{(AB}{}^{EL}\Omega_{CD)EL}
+\tfrac{3}{2} \xi ^{ELFH} \Omega_{E(ABC} \Omega_{D)LFH}\nonumber\\
&-2  \Omega_{E(BC}{}^H \xi_A{}^{ELF} \Omega_{D)LFH} 
+\tfrac{1}{4}\xi^{ELFH}\Omega_{ABCD}\Omega_{ELFH}
-\tfrac{1}{3} K \xi  \Omega_{ABCD}\nonumber \\
&+\tfrac{1}{2}\xi_{(AB}{}^{EL}\Omega_{CD)}{}^{FH}\Omega_{ELFH} 
+\tfrac{2}{5} \xi_{(CD} \nabla _{AB)}K +\tfrac{12}{5} \xi_{E(BCD} \nabla_{A)}{}^EK \nonumber\\
&-3 \Omega_{E(BCD} \nabla_{A)}{}^E\xi 
-\tfrac{3}{2}\Omega_{(A}{}^{ELF}\nabla_{CD}\xi_{B)ELF}
-\tfrac{3}{2} \Omega_{F(A}{}^{EL} \nabla_D{}^F\xi_{BC)EL}\nonumber\\
&-\tfrac{9}{2} \Omega_{(AB}{}^{EL} \nabla_D{}^F\xi_{C)ELF} 
-\tfrac{9}{2} \nabla_{L(D}\nabla_C{}^E\xi_{AB)E}{}^L
-\tfrac{3}{2} \nabla_{L(D}\nabla^{EL}\xi_{ABC)E}\nonumber\\
&-6 K \nabla_{E(D}\xi_{ABC)}{}^E 
+3 \Omega_{L(AB}{}^E \nabla^{LF}\xi_{CD)EF}
-3 \Omega_{(ABC}{}^E \nabla^{LF}\xi_{D)ELF}\nonumber\\ 
&-3\kappa^{EL}\nabla_{L(D}\Psi_{ABC)E} 
+3\kappa_{(A}{}^E \nabla_D{}^L\Psi_{BC)EL}.\label{nabla2xi0c}
\end{align}
\end{subequations}

The equations presented in this section 
have been deduced using the tensor algebra
suite {\tt xAct} for {\tt Mathematica} ---see \cite{xAct}.

\subsection{Propagation of the Killing spinor equation}

A straightforward consequence of the Killing spinor equation 
\eqref{KillingSpinorEquation} in a vacuum spacetime is that:
\begin{equation}\label{boxkappa}
\square\kappa_{AB}=-\Psi_{ABCD}\kappa^{CD},
\end{equation}
where $\square \equiv \nabla^{AA'}\nabla_{AA'}$.  The latter equation
is obtained by applying the differential operator $\nabla^{AA'}$ to
equation \eqref{KillingSpinorEquation} and then using the vacuum commutator relation for the spacetime Levi-Civita connection.

\bigskip
The wave equation \eqref{boxkappa} plays a role in the
discussion of the \emph{propagation} of the Killing spinor
equation. More precisely, one has the following result
---cfr. \cite{GarVal08a} for further details.

\begin{lemma}
Let $\kappa_{AB}$ be a solution to equation \eqref{boxkappa}. Then the
corresponding spinor fields $H_{A'ABC}$ and $S_{AA'BB'}$ will satisfy
the system of wave equations
\begin{subequations}
\begin{eqnarray}
&& \square H_{A'ABC}= 4\left(\Psi_{(AB}{}^{PQ} H_{C)PQA'} + \nabla_{(A}{}^{Q'}S_{BC)Q'A'}\right), \label{wave1} \\
&& \square S_{AA'BB'} = -\nabla_{AA'} \left( \Psi_B{}^{PQR}H_{B'PQR} \right)-\nabla_{BB'}\left( \Psi_A{}^{PQR}H_{A'PQR}\right) \nonumber \\
&& \hspace{4cm} + 2\Psi_{AB}{}^{PQ}S_{PA'QB'} + 2 \bar{\Psi}_{A'B'}{}^{P'Q'}S_{AP'BQ'}. \label{wave2}
\end{eqnarray}
\end{subequations}
\end{lemma}

The crucial observation is that the right hand sides of equations
\eqref{wave1} and \eqref{wave2} are homogeneous
expressions of the unknowns and their first order derivatives. The
hyperbolicity of equations \eqref{wave1} and \eqref{wave2} imply the
following result ---again, cfr. \cite{GarVal08a} for further details.

\begin{proposition}
\label{Proposition:KSDevelopment}
The development $(\mathcal{M},g_{\mu\nu})$ of an
initial data set for the vacuum Einstein field equations,
$(\mathcal{S},h_{ab},K_{ab})$, has a Killing spinor in the 
domain of dependence of $\mathcal{U}\subset\mathcal{S}$ if 
and only if the following equations are satisfied on $\mathcal{U}$.
\begin{subequations} 
\begin{eqnarray}
&& H_{A'ABC}=0, \label{old_kspd1}\\
&& \nabla H_{A'ABC}=0, \label{old_kspd2}\\
&& S_{AA'BB'}=0, \label{old_kspd3}\\
&& \nabla S_{AA'BB'}=0. \label{old_kspd4}
\end{eqnarray}
\end{subequations}
\end{proposition}

\subsection{The Killing spinor data equations}

The \emph{Killing spinor data} conditions obtained in Proposition
\ref{Proposition:KSDevelopment} can be reexpressed in terms of
conditions on the spinor $\kappa_{AB}$ which are intrinsic to the
hypersurface $\mathcal{S}$. For this one uses the split of $\xi_{AA'}$
and $H_{A'ABC}$ given by equations \eqref{Split:xi}-\eqref{Split:H}.
Extensive computations using the {\tt xAct} suite for {\tt Mathematica} render the
following result. 

\begin{theorem}
\label{Theorem:KSData}
Let $(\mathcal{S},h_{ab},K_{ab})$ be an initial data set for the Einstein vacuum field equations, where $\mathcal{S}$ is a Cauchy hypersurface. Let $\mathcal{U}\subset\mathcal{S}$ be an open set. 
The development of the initial data set will then have a Killing spinor in the domain of dependence of $\mathcal{U}$ if and only if 
\begin{subequations}
\begin{eqnarray}
&& \xi_{ABCD}=0,\label{kspd1}\\
&& \Psi_{(ABC}{}^F\kappa_{D)F}=0, \label{kspd2}\\
&& 3\kappa_{(A}{}^E\nabla_B{}^F\Psi_{CD)EF}+\Psi_{(ABC}{}^F\xi_{D)F}=0,\label{kspd3}
\end{eqnarray}
\end{subequations}
are satisfied on $\mathcal{U}$. The Killing spinor is obtained by 
evolving \eqref{boxkappa} with initial data satisfying conditions \eqref{kspd1}-\eqref{kspd3} and
\begin{equation}
\nabla\kappa_{AB}=-\tfrac{2}{3}\xi_{AB} \label{kspd4}
\end{equation}
on $\mathcal{U}$.
\end{theorem}

\bigskip
\noindent
\textbf{Remark 1.} Conditions \eqref{kspd1}-\eqref{kspd3} are intrinsic
to $\mathcal{U}\subset \mathcal{S}$ and will be referred to as the
\emph{Killing spinor initial data equations}. In particular, equation
\eqref{kspd1}, which can be written as
\begin{equation}
\label{SpatialKillingSpinorEquation}
\nabla_{(AB}\kappa_{CD)}=0,
\end{equation}
will be called the \emph{spatial Killing spinor equation}, whereas \eqref{kspd2} and \eqref{kspd3} will be known as
the \emph{algebraic conditions}.

\bigskip
\noindent
\textbf{Remark 2.} Theorem \ref{Theorem:KSData} is an improvement on
Proposition 6 of \cite{GarVal08a} where the interdependence of the equations implied by \eqref{old_kspd1}-\eqref{old_kspd4} was not analysed.

\begin{proof}
The proof of Theorem \ref{Theorem:KSData} consists of a space spinor
decomposition of the conditions \eqref{old_kspd1}-\eqref{old_kspd4}
and of an analysis of the dependencies of the resulting conditions.
All calculations are made on $\mathcal{U}\subset\mathcal{S}$.

\begin{itemize}
\item \emph{Decomposition of equation \eqref{old_kspd1}.} Splitting
$\tau_F{}^{A'}H_{A'ABC}$ into irreducible parts gives that
\eqref{old_kspd1} is equivalent to
\begin{subequations}
\begin{eqnarray}
&&\xi_{ABCD}=0,\label{xi4vanishing}\\
&&\nabla\kappa_{AB}=-\tfrac{2}{3}\xi_{AB}\label{nablakappa}.
\end{eqnarray}
\end{subequations}

\item \emph{Decomposition of equation \eqref{old_kspd2}.} It follows that
\[
\tau_D{}^{A'}\nabla H_{A'ABC}=\nabla(\tau_D{}^{A'}
H_{A'ABC})+H_{A'ABC}K_{DF}\tau^{FA'}.
\]
Hence, under the condition \eqref{old_kspd1}, the irreducible parts of
$\tau_D{}^{A'}\nabla H_{A'ABC}$ are given by
\begin{subequations}
\begin{eqnarray}
&& \nabla\xi_{ABCD}=0,\label{nablaxi4}\\
&& \nabla^2\kappa_{AB}=-\tfrac{2}{3}\nabla\xi_{AB}.\label{nabla2kappa}
\end{eqnarray}
\end{subequations}
 From the commutator \eqref{commutator1} together with \eqref{xi4vanishing} and \eqref{nablakappa} we get
\begin{align*}
\nabla\xi_{ABCD}={}&\nabla\nabla_{(AB}\kappa_{CD)}\\
={}&2\Psi_{(ABC}{}^F\kappa_{D)F}-\tfrac{1}{3}\Omega_{(AB}\xi_{CD)}-\tfrac{1}{3}\Omega_{ABCD}\xi+\tfrac{2}{3}\Omega_{(ABC}{}^F\xi_{D)F}-\tfrac{2}{3}\nabla_{(AB}\xi_{CD)}.
\end{align*}
Equation \eqref{Dxi3} and again \eqref{xi4vanishing} and \eqref{nablakappa} then yield
\begin{equation}
\nabla\xi_{ABCD}=4\Psi_{(ABC}{}^F\kappa_{D)F}.
\end{equation}
Using the commutator \eqref{commutator1} one obtains that
\begin{subequations}
\begin{align}
\nabla\xi={}&
\nabla_{AB}\nabla\kappa^{AB}
-\tfrac{1}{3}K\xi
+\tfrac{2}{3}K^{AB}\xi_{AB}
+\tfrac{2}{3}\Omega^{AB}\xi_{AB}
-\Omega^{ABCD}\xi_{ABCD}
-\tfrac{1}{2}K^{AB}\nabla\kappa_{AB}
\label{nablaxi0a}\\
\nabla\xi_{AB}={}&
\tfrac{3}{2}\Psi_{ABCD}\kappa^{CD}
-\tfrac{1}{2}K_{AB}\xi
-\tfrac{1}{3}K\xi_{AB}
+\tfrac{1}{2}K^C{}_{(A}\xi_{B)C}
+\tfrac{3}{4}K^{CD}\xi_{ABCD}
-\tfrac{1}{2}\xi\Omega_{AB} \nonumber \\
&-\tfrac{1}{2}\xi^C{}_{(A}\Omega_{B)C}
+\tfrac{3}{4}\Omega^{CD}\xi_{ABCD}
+\tfrac{3}{2}\xi_{(A}{}^{CDF}\Omega_{B)CDF}
+\tfrac{1}{2}\xi^{CD}\Omega_{ABCD} \nonumber \\
&-\tfrac{3}{4}K^C{}_{(A}\nabla\kappa_{B)C}
+\nabla_{C(A}\nabla\kappa_{B)}{}^C\label{nablaxi2a}
\end{align}
\end{subequations}

In terms of the normal derivative and the Sen connection, equation \eqref{boxkappa} reads 
\begin{align}
\nabla^2\kappa_{AB}={}&
-2\Psi_{ABCD}\kappa^{CD}
 -K\nabla\kappa_{AB}
 -\tfrac{2}{3}\nabla_{AB}\xi
 -\tfrac{4}{3}\nabla_{C(A}\xi^C{}_{B)}
 -2\nabla^{CD}\xi_{ABCD} \nonumber \\
&+\tfrac{1}{3}K_{AB}\xi -\tfrac{2}{3}K^C{}_{(A}\xi_{B)C}
 +K^{CD}\xi_{ABCD}
 +\tfrac{2}{3}\Omega_{AB}\xi
 +\tfrac{4}{3}\xi^C{}_{(A}\Omega_{B)C} \nonumber \\
&+2\xi_{ABCD}\Omega^{CD}.\label{boxkappasplit}
\end{align}
It is worth stressing that equations \eqref{nablaxi0a},
\eqref{nablaxi2a} and \eqref{boxkappasplit} are valid not only on
$\mathcal{U}$, but on the spacetime. Hence, it makes sense taking
normal derivatives of these equations. Using \eqref{nablaxi2a},
\eqref{nablakappa} and \eqref{xi4vanishing}, the wave equation
\eqref{boxkappa} is seen to imply
\begin{align*}
\nabla^2\kappa_{AB}+\tfrac{2}{3}\nabla\xi_{AB}={}&
-\Psi_{ABCD}\kappa^{CD}
+\tfrac{4}{9}K\xi_{AB}
+\tfrac{1}{3}\Omega_{AB}\xi
+\xi^C{}_{(A}\Omega_{B)C}\nonumber \\
&+\tfrac{1}{3}\Omega_{ABCD}\xi^{CD}
-\tfrac{2}{3}\nabla_{AB}\xi
-\tfrac{2}{3}\nabla_{C(A}\xi^C{}_{B)}.
\end{align*}
Using equations \eqref{Dxi2}, \eqref{nablakappa},
\eqref{xi4vanishing}, the latter equation reduces to \eqref{nabla2kappa}.
This far we have that for all solutions to \eqref{boxkappa}, the system \eqref{old_kspd1}, \eqref{old_kspd2} is equivalent to the system \eqref{kspd1}, \eqref{kspd2}, \eqref{kspd4}.

\item \emph{Decomposition of equation \eqref{old_kspd3}.} Splitting
$\tau_C{}^{A'}\tau_D{}^{B'}S_{AA'BB'}$ into irreducible parts yields
\begin{subequations}
\begin{align}
&\nabla_{(AB}\nabla\kappa_{CD)}-\Omega_{ABCD}\xi+\tfrac{4}{3}K_{(ABC}{}^F\xi_{D)F}-K_{(ABC}{}^F\nabla\kappa_{D)F}-\tfrac{4}{3}\nabla_{(AB}\xi_{CD)}=0,\label{S0Full1}\\
&2\nabla\xi-\tfrac{4}{3}K^{AB}\xi_{AB}+K^{AB}\nabla\kappa_{AB}=0, \label{S0Full2}\\
&\tfrac{4}{3}\nabla_{AB}\xi^{AB}+K\xi-\tfrac{4}{3}\Omega^{AB}\xi_{AB}+\Omega^{AB}\nabla\kappa_{AB}-\nabla_{AB}\nabla\kappa^{AB}=0, \label{S0Full3}
\\
&\tfrac{1}{2}K_{BD}\xi
-\tfrac{2}{3}K^A{}_{(B}\xi_{D)A}
+\tfrac{1}{2}K^A{}_{(B}\nabla\kappa_{D)A}
-\tfrac{2}{3}K_{BDAC}\xi^{AC}
+\tfrac{1}{2}K_{BDAC}\nabla\kappa^{AC}\nonumber\\
&+\tfrac{2}{3}\nabla\xi_{BD}
-\tfrac{1}{2}\nabla^2\kappa_{BD}
+\nabla_{BD}\xi =0. \label{S0Full4}
\end{align}
\end{subequations}
Using equations \eqref{Dxi1}, \eqref{Dxi3}, \eqref{xi4vanishing},
\eqref{nablakappa} and \eqref{nabla2kappa}, one sees that equations
\eqref{S0Full1}-\eqref{S0Full3} simplify to
\begin{subequations}
\begin{eqnarray}
&& \Psi^F{}_{(ABC}\kappa_{D)F}=0,\\
&& \nabla\xi=K^{AB}\xi_{AB},\label{nablaxi0}\\
&& \nabla\xi_{BD}=-\tfrac{1}{2}K_{BD}\xi
+K^A{}_{(B}\xi_{D)A}
+K_{BDAC}\xi^{AC}
-\nabla_{BD}\xi,\label{nablaxi2}
\end{eqnarray}
\end{subequations}
while equation \eqref{S0Full4} is seen to be satisfied
identically. Furthermore, employing equations \eqref{Dxi1},
\eqref{Dxi2}, \eqref{nablaxi0a}, \eqref{xi4vanishing},
\eqref{nablakappa} and \eqref{nablaxi2a} one obtains equation
\eqref{nablaxi0} and \eqref{nablaxi2}. Hence, they are a consequence
of the commutators, \eqref{nablakappa} and \eqref{xi4vanishing}. One
concludes that for all solutions to \eqref{boxkappa}, the equations
\eqref{kspd1}, \eqref{kspd2} together with \eqref{kspd4} are equivalent to \eqref{old_kspd1}, \eqref{old_kspd2}, \eqref{old_kspd3}.

\item \emph{Decomposition of equation \eqref{old_kspd4}}. A straightforward
computation shows that
\begin{eqnarray*}
&& \tau_C{}^{A'}\tau_D{}^{B'}\nabla S_{AA'BB'} \\
&& \hspace{1cm} =\nabla(\tau_C{}^{A'}\tau_D{}^{B'}S_{AA'BB'})+K_{CF}S_{AA'BB'}\tau_D{}^{B'}\tau^{FA'}+K_{DF}S_{AA'BB'}\tau_C{}^{A'}\tau^{FB'}.
\end{eqnarray*}
Hence, if condition \eqref{old_kspd3} holds, the irreducible parts of
$\tau_C{}^{A'}\negthinspace\tau_D{}^{B'}\nabla S_{AA'BB'}$ are $\nabla$-derivatives
of \eqref{S0Full1}-\eqref{S0Full4}. Using equation
\eqref{nabla2kappa}, these components become
\begin{subequations}
\begin{eqnarray}
&&
\Omega_{ABCD}\nabla\xi
+\Omega_{(AB}\nabla\xi_{CD)}
-2\Omega^F{}_{(ABC}\nabla\xi_{D)F}
+\tfrac{2}{3}\xi_{(AB}\nabla\Omega_{CD)}\nonumber\\
&&\hspace{1cm}-\tfrac{1}{2}\nabla\kappa_{(AB}\nabla\Omega_{CD)}
+\xi\nabla\Omega_{ABCD}
+\tfrac{4}{3}\xi^F{}_{(A}\nabla\Omega_{BCD)F}\nonumber\\
&&\hspace{1cm}-(\nabla\kappa^F{}_{(A})\nabla\Omega_{BCD)F}
+\tfrac{4}{3}\nabla\nabla_{(AB}\xi_{CD)}
-\nabla\nabla_{(AB}\nabla\kappa_{CD)}=0, \label{nablaS0Full1}\\
&&2\nabla^2\xi-2K^{AB}\nabla\xi_{AB}+(\nabla K^{AB})\nabla\kappa_{AB}-\tfrac{4}{3}\xi^{AB}\nabla K_{AB}=0, \label{nablaS0Full2}\\
&&\xi\nabla K
+K\nabla\xi
-2\Omega^{AB}\nabla\xi_{AB}
-\tfrac{4}{3}\xi^{AB}\nabla\Omega_{AB}
+(\nabla\kappa^{AB})\nabla\Omega_{AB}\nonumber\\
&&\hspace{1cm}+\tfrac{4}{3}\nabla\nabla_{AB}\xi^{AB}
-\nabla\nabla_{AB}\nabla\kappa^{AB}=0, \label{nablaS0Full3}\\
&&\nabla^3\kappa_{BD}+\tfrac{2}{3}\nabla^2\xi_{BD}
=\tfrac{4}{3}\xi^A{}_{(B}\nabla\kappa_{D)A}
+\xi\nabla\kappa_{BD}
-\tfrac{4}{3}\xi^{AC}\nabla K_{BDAC}  \nonumber \\
&& \hspace{1cm}+(\nabla K^A{}_{(B})\nabla\kappa_{D)A}+(\nabla K_{BDAC})\nabla\kappa^{AC}
+K_{BD}\nabla\xi -2 K^A{}_{(B}\nabla\xi_{D)A}\nonumber \\
&& \hspace{1cm}-2 K_{BDAC}\nabla^{AC}+2\nabla^2\xi_{BD}
+2\nabla\nabla_{BD}\xi. \label{nablaS0Full4}
\end{eqnarray}
\end{subequations}
Now, using the commutator \eqref{commutator1}, and equations
\eqref{nabla2kappa} and \eqref{nablakappa} it is easy so see that
\begin{equation}\label{nablaDdnablakappa}
\nabla\nabla_{AB}\nabla\kappa_{CD}=-\tfrac{2}{3}\nabla\nabla_{AB}\xi_{CD}.
\end{equation}
Taking the normal derivative of the spacetime equations
\eqref{nablaxi0a}-\eqref{nablaxi2a} and using the relations
\eqref{nablaDdnablakappa}, \eqref{Dxi1}, \eqref{Dxi2},
\eqref{xi4vanishing}, \eqref{nablakappa}, \eqref{nablaxi4} and
\eqref{nabla2kappa} one gets
\begin{align*}
\nabla^2\xi={}&\xi^{AB}\nabla K_{AB}+K^{AB}\nabla\xi_{AB},\\
\nabla^2\xi_{AB}={}&
-\tfrac{1}{2}\xi\nabla K_{AB}
-\xi^C{}_{(A}\nabla K_{B)C}
+\tfrac{1}{3}\xi_{AB}\nabla K
-\tfrac{1}{2}K_{AB}\nabla\xi
+\tfrac{1}{3}K \nabla\xi_{AB} \nonumber \\
&+K^C{}_{(A}\nabla\xi_{B)C}-\Omega^C{}_{(A}\nabla\xi_{B)C}
+\Omega_{ABCD}\nabla\xi^{CD}
+\xi^C{}_{(A}\nabla\Omega_{B)C} \nonumber \\
&+\xi^{CD}\nabla\Omega_{ABCD}
-\nabla\nabla_{AB}\xi.
\end{align*}
Using these last two equations together with equations \eqref{Dxi1},
\eqref{Dxi3}, \eqref{xi4vanishing}, \eqref{nablakappa},
\eqref{nablaxi4}, \eqref{nabla2kappa} and \eqref{nablaxi0} one finds
that the system \eqref{nablaS0Full1}-\eqref{nablaS0Full4} reduces to
\begin{subequations}
\begin{eqnarray}
&&
4\Psi^F{}_{(ABC}\xi_{D)F}+6\kappa^F{}_{(A}\nabla\Psi_{BCD)F}\label{TempEqnablaPsi}=0, \\
&&\nabla^3\kappa_{BD}+\tfrac{2}{3}\nabla^2\xi_{BD}=0.\label{nabla3kappa}
\end{eqnarray}
\end{subequations}
Taking the normal derivative of equation \eqref{boxkappasplit} and
using equations \eqref{Dxi2}, \eqref{xi4vanishing},
\eqref{nablakappa}, \eqref{nablaxi4}, \eqref{nabla2kappa} and
\eqref{nablaxi0} one gets equation \eqref{nabla3kappa}. Finally, using
the Bianchi equation \eqref{Bianchi1}, one has that equation
\eqref{TempEqnablaPsi} reduces to
\begin{equation}
3\kappa_{(A}{}^E\nabla_B{}^F\Psi_{CD)EF}
+\Psi_{(ABC}{}^F\xi_{D)F}=0
\end{equation}
\end{itemize}
This completes the proof.
\end{proof}

\noindent
\textbf{Remark.} Note that the result is independent of $K_{AB}$ and $\Omega_{AB}$.

\subsubsection{The Killing spinor initial data conditions in terms of the Levi-Civita connection}
It should be stressed that the Killing spinor equations
\eqref{kspd1}-\eqref{kspd3} are truly intrinsic to the hypersurface
$\mathcal{S}$. This can be more easily seen by expressing the Sen
connection, $\nabla_{AB}$, in terms of the intrinsic (Levi-Civita)
connection of the hypersurface, $D_{AB}$, and the second fundamental
form $K_{ABCD}$. One obtains the following completely equivalent set
of equations:
\begin{eqnarray*}
&&D_{(AB}\kappa_{CD)}+\Omega_{(ABC}{}^E\kappa_{D)E}=0,\\
&&\Psi_{(ABC}{}^F\kappa_{D)F}=0,\\
&&3\kappa_{(A}{}^E D_B{}^F\Psi_{CD)EF}
-\tfrac{3}{4}\Psi_{L(ABC}D^{HL}\kappa_{D)H}
-\tfrac{3}{4}\Psi_{L(ABC}D_{D)}{}^F\kappa^L{}_F\nonumber\\
&&\hspace{1cm}+\tfrac{3}{4}\Psi_{(ABC}{}^L\Omega_{D)FHL}\kappa^{FH}
+\tfrac{3}{2}\Psi_{(AB}{}^{HL}\kappa_{C}{}^{F}\Omega_{D)FHL}
-\tfrac{3}{2}\Psi_{FH(A}{}^{L}\Omega_{BCD)L}\kappa^{FH}\nonumber\\
&&\hspace{1cm}+\tfrac{3}{8}\Psi_{FH(AB}\kappa_{CD)}\Omega^{FH}
+\tfrac{3}{4}\Psi_{FH(AB}\Omega_{CD)}\kappa^{FH}=0,
\end{eqnarray*}
where the last expression was simplified using the first algebraic
condition, and the value of the Weyl spinor is expressed
in terms of initial data quantities via formulae
\eqref{Weyl:Electric}-\eqref{Weyl:Magnetic}.

\subsection{The integrability conditions of the spatial Killing spinor equation}
For the rest of the paper we assume that the tensor $K_{ab}$
is symmetric ---accordingly, $\Omega_{AB}=0$. 
The condition $\xi_{ABCD}\equiv\nabla_{(AB}\kappa_{CD)}=0$ 
does not immediately give information about the other irreducible
components of $\nabla_{AB}\kappa_{CD}$, namely $\xi$ and
$\xi_{AB}$. However, using $\xi_{ABCD}=0$ and $\Omega_{AB}=0$ in the
relations \eqref{Dxi1}-\eqref{Dxi3} one finds that
$\nabla_{AB}\xi_{CD}$ can be written in terms of $\nabla_{AB}\xi$ and
lower order derivatives of $\kappa_{AB}$. Furthermore, using
$\xi_{ABCD}=0$ in the relations \eqref{nabla2xi0a}-\eqref{nabla2xi0c},
we see that the second order derivatives of $\xi$ can be expressed in
terms of lower order derivatives of $\kappa_{AB}$.  This yields the
following result which will play a role in the sequel:

\begin{lemma}
\label{Lemma:ThirdDerivative}
Assume that $\nabla_{(AB}\kappa_{CD)}=0$, then
\[
\nabla_{AB}\nabla_{CD} \nabla_{EF} \kappa_{GH} = H_{ABCDEFGH},
\]
where $H_{ABCDEFGH}$ is a linear combination of $\kappa_{AB}$,
$\nabla_{AB}\kappa_{CD}$ and $\nabla_{AB}\nabla_{CD} \kappa_{EF}$ with
coefficients depending on $\Psi_{ABCD}$, $\hat{\Psi}_{ABCD}$ and
$K_{ABCD}$.
\end{lemma}

\medskip
\noindent
\textbf{Remark.} It is important to point out that the assertion of
the lemma is false if $\nabla_{(AB}\kappa_{CD)}\neq 0$.

\section{The approximate Killing spinor equation}
\label{Section:ApproximateKS}

In what follows we will regard the spatial Killing spinor equation
\eqref{kspd1} as the key condition of the Killing spinor initial data
equations. Equation (\ref{kspd1}) is an overdetermined
condition for the 3 (complex) components of the spinor
$\kappa_{AB}$: not every initial data set $(\mathcal{S},h_{ab},K_{ab})$
admits a solution. One would like to deduce a new equation which always
has a solution and such that any solution to equation (\ref{kspd1}) is
also a solution to the new equation. 

\subsection{The approximate Killing spinor operator}
Let $\mathfrak{S}_2$ and $\mathfrak{S}_4$ denote, respectively, the
spaces of totally symmetric valence 2 and valence 4
spinors. Given $\zeta_{ABCD}, \; \chi_{ABCD}\in \mathfrak{S}_4$, we
introduce an inner product in $\mathfrak{S}_4$ via:
\[
\langle \zeta_{ABCD}, \chi_{EFGH}\rangle = \int_{\mathcal{S}} \zeta_{ABCD} \hat{\chi}^{ABCD} \mbox{d}\mu,
\]
where $\mbox{d}\mu$ denotes the volume form of the 3-metric
$h_{ab}$. We introduce the spatial Killing spinor operator $\Phi$
via
\[
\Phi: \mathfrak{S}_2 \rightarrow \mathfrak{S}_4, \quad \Phi(\kappa)_{ABCD}= \nabla_{(AB}\kappa_{CD)}. 
\]
Now, consider the pairing
\begin{eqnarray*}
&& \langle \nabla_{(AB}\kappa_{CD)}, \zeta_{EFGH} \rangle = \int_{\mathcal{S}} \nabla_{(AB}\kappa_{CD)} \hat{\zeta}^{ABCD} \mbox{d}\mu \\
&& \phantom{\langle \nabla_{(AB}\kappa_{CD)}, \zeta_{EFGH} \rangle}= \int_{\mathcal{S}} \nabla_{AB}\kappa_{CD} \hat{\zeta}^{ABCD} \mbox{d}\mu.
\end{eqnarray*}
The formal adjoint, $\Phi^*$, of the spatial Killing operator can be obtained from the latter expression by integration by
parts. To this end we note the identity:
\begin{eqnarray}
&& \int_{\mathcal{U}} \nabla^{AB} \kappa^{CD} \hat{\zeta}_{ABCD} \mbox{d}\mu - \int_{\mathcal{U}} \kappa^{AB} \widehat{\nabla^{CD} \zeta_{ABCD}}\mbox{d}\mu + \int_{\mathcal{U}} 2\kappa^{AB} \Omega^{CDF}{}_A\hat\zeta_{BCDF}\mbox{d}\mu \nonumber \\
&& \hspace{2cm}= \int_{\partial \mathcal{U}} n^{AB} \kappa^{CD} \hat{\zeta}_{ABCD} \mbox{d}S, \label{IntegrationbyParts}
\end{eqnarray}
with $\mathcal{U}\subset \mathcal{S}$, and where $\mbox{d}S$ denotes
the area element of $\partial \mathcal{U}$, $n_{AB}$ is the spinorial
counterpart of its outward pointing normal, and $\zeta_{ABCD}$ is a
symmetric spinor. From \eqref{IntegrationbyParts} it follows that 
\begin{equation}
\label{FormalAdjoint}
\Phi^*:\mathfrak{S}_4 \rightarrow \mathfrak{S}_2, \quad \Phi^*(\zeta)_{CD}=\nabla^{AB}\zeta_{ABCD}-2\Omega^{ABF}{}_{(C}\zeta_{D)ABF}.
\end{equation}

\begin{definition}
The composition operator $L\equiv \Phi^*\circ \Phi: \mathfrak{S}_2\rightarrow \mathfrak{S}_2$ given by:
\begin{equation}
L(\kappa_{CD}) \equiv \nabla^{AB} \nabla_{(AB} \kappa_{CD)}-\Omega^{ABF}{}_{(A}\nabla_{|DF|}\kappa_{B)C}-\Omega^{ABF}{}_{(A}\nabla_{B)F}\kappa_{CD}=0, \label{ApproximateKillingEquation} 
\end{equation}
will be called the \emph{approximate Killing spinor operator}, and equation
\eqref{ApproximateKillingEquation} the \emph{approximate Killing spinor equation}.
\end{definition}

\medskip
\noindent
\textbf{Remark.} Note that every solution to the spatial Killing spinor
equation \eqref{SpatialKillingSpinorEquation} is also a solution to
equation \eqref{ApproximateKillingEquation}.

\subsection{Ellipticity of the approximate Killing spinor operator}
As a prior step to the analysis of the solutions to the approximate Killing spinor equation \eqref{ApproximateKillingEquation}, we look first at its ellipticity properties.

\begin{lemma}
The operator $L$ defined by equation \eqref{ApproximateKillingEquation} is a formally self-adjoint elliptic operator.
\end{lemma}

\begin{proof} 
The operator is by construction formally self-adjoint 
as it is given by the composition of an operator and its formal adjoint. In
order to verify ellipticity, it suffices to look at the operator
\[
L'(\kappa)_{\mathbf{CD}}\equiv \partial^{\mathbf{AB}} \partial_{(\mathbf{AB}} \kappa_{\mathbf{CD})},
\]
corresponding to the principal part of $L$ in some Cartesian spin
frame. In the corresponding Cartesian coordinates $(x^1,x^2,x^3)$ one
has that
\[
\partial_{\mathbf{AB}}= \frac{1}{\sqrt{2}}
\left(
\begin{array}{cc}
-\partial_1 - \mbox{i}\partial_2 & \partial_3 \\
\partial_3 & \partial_1 -\mbox{i}\partial_2
\end{array}
\right),
\quad 
\partial^{\mathbf{AB}}= \frac{1}{\sqrt{2}}
\left(
\begin{array}{cc}
-\partial_1 + \mbox{i}\partial_2 & \partial_3 \\
\partial_3 & \partial_1 +\mbox{i}\partial_2
\end{array}
\right).
\]
In particular,  $\partial^{\mathbf{AB}} \partial_{\mathbf{AB}}=\Delta\equiv \partial^2_1 +\partial^2_2 +\partial^2_3$, the flat Laplacian. One notes that
\[
\partial^{\mathbf{PQ}}\partial_{(\mathbf{PQ}} \kappa_{\mathbf{AB})} = \frac{1}{6} \partial^{\mathbf{PQ}}\partial_{\mathbf{PQ}} \kappa_{\mathbf{AB}}+ \frac{2}{3} \partial^{\mathbf{PQ}} \partial_{\mathbf{P}(\mathbf{A}} \kappa_{\mathbf{B})\mathbf{Q}} + \frac{1}{6} \partial^{\mathbf{PQ}} \partial_{\mathbf{AB}} \kappa_{\mathbf{PQ}}. 
\]
Now, writing 
\[
\kappa_{0}\equiv\kappa_{00}, \quad \kappa_{1}\equiv\kappa_{01}, \quad \kappa_{2}\equiv\kappa_{11},
\]
one has that $L'$ can be expressed in matricial form as $A^{ij} \partial_i \partial_j u$,
where 
\begin{equation}
\label{matrixA}
A^{ij}\partial_i \partial_j \equiv 
\frac{1}{12}
\left(
\begin{array}{cccccc}
7\Delta -\partial_3^2 & -2\partial_1 \partial_3 & \partial_2^2-\partial_1^2 & 0 & -2\partial_2 \partial_3 & -2\partial_1 \partial_2 \\
-\partial_1\partial_3 & 6\Delta+2\partial^2_3 & \partial_1 \partial_3 & \partial_2\partial_3 & 0 & \partial_2 \partial_3 \\
\partial_2^2-\partial_1^2 & 2\partial_1\partial_3 & 7\Delta-\partial^2_3 & 2 \partial_1\partial_2 & -2\partial_2\partial_3 & 0 \\
0 & 2\partial_2\partial_3 & 2\partial_1\partial_2 & 7\Delta-\partial_3^2 & -2\partial_1\partial_3 & \partial^2_2-\partial_1^2 \\
-\partial_2\partial_3 & 0 & -\partial_2\partial_3 & -\partial_1\partial_3 & 6\Delta+2\partial_3^2 & \partial_1\partial_3 \\
-2\partial_1\partial_2 & \partial_2\partial_3 & 0& \partial_2^2-\partial_1^2 & 2 \partial_1\partial_3 & 7\Delta-\partial_3^2
\end{array}
\right),
\end{equation}
and 
\begin{equation}
\label{vectorU}
u\equiv 
\left(
\begin{array}{c}
\mbox{Re}(\kappa_0) \\
\mbox{Re}(\kappa_1) \\
\mbox{Re}(\kappa_2) \\
\mbox{Im}(\kappa_0) \\
\mbox{Im}(\kappa_1) \\
\mbox{Im}(\kappa_2) \\
\end{array}
\right).
\end{equation}
The symbol, $l(\xi_i)$, of the operator given by \eqref{matrixA} is
then given by replacing $\partial_i$ with $\xi_i\in \Real^3$. One
finds that
\[ \det l(\xi_i)=\frac{1}{36}\left(
(\xi_1)^2+(\xi_2)^2+(\xi_3)^2\right)^6,
\] so that $\det l(\xi_i)=0$ if and only if $\xi_i=0$. Accordingly,
the operator $L=\Phi^*\circ\Phi$ is elliptic.
\end{proof}

\subsection{A variational formulation}
We note that the approximate Killing spinor
equation arises naturally from a variational principle.

\begin{lemma}
The approximate Killing spinor equation \eqref{ApproximateKillingEquation} is the Euler-Lagrange equation of the functional
\begin{equation}
J = \int_{\mathcal{S}} \nabla_{(AB} \kappa_{CD)} \widehat{\nabla^{AB} \kappa^{CD}} \mbox{d}\mu. \label{functional}
\end{equation}
\end{lemma}

\begin{proof}
This is a direct consequence of the identity \eqref{IntegrationbyParts}. 
\end{proof}

\section{Asymptotically Euclidean manifolds}
\label{Section:AsymptoticallyEuclideanData}

After having studied some formal properties of the Killing spinor
initial data equations \eqref{kspd1}-\eqref{kspd3},\eqref{kspd4}, and the approximate Killing spinor equation 
\eqref{ApproximateKillingEquation}, we proceed to analyse their
solvability on asymptotically Euclidean manifolds. In order to do this we introduce some relevant terminology and ancillary results.

\subsection{General assumptions}
In what follows, we will be concerned with vacuum spacetimes arising
as the development of 
asymptotically Euclidean data sets. Let $(\mathcal{S},h_{ab},K_{ab})$,
denote a smooth initial data set for the vacuum Einstein field
equations. The pair $(h_{ab},K_{ab})$ satisfies on the 3-dimensional
manifold $\mathcal{S}$  the vacuum constraint
equations
\begin{subequations}
\begin{eqnarray}
&& -2r - K^a{}_a K^b{}_b + K_{ab}K^{ab}=0, \label{Hamiltonian}\\
&& D^a K_{ab}- D_b K^a{}_a=0, \label{Momentum}
\end{eqnarray}
\end{subequations}
where $r$ and $D$ denote, respectively, the Ricci scalar and the
Levi-Civita connection of the negative definite 3-metric $h_{ab}$,
while $K_{ab}$ corresponds to the extrinsic curvature of
$\mathcal{S}$. The unusual coefficients in the formulae above come
from our normalisation of $\tau^\mu$.  For an \emph{asymptotic end} it
will be understood an open set diffeomorphic to the complement of a
closed ball in $\Real^3$.  In what follows, the 3-manifold
$\mathcal{S}$ will be assumed to be the union of a compact set and two
asymptotically Euclidean ends, $i_1, \;i_2$.

\subsection{Weighted Sobolev norms}
In order to discuss the decays of the various fields on the 3-manifold
$\mathcal{S}$ we make use of weighted Sobolev spaces.
In what follows, we follow the ideas of \cite{Can81} written in
terms of the conventions of \cite{Bar86}. Choose an
arbitrary point $O\in \mathcal{S}$, and let
\[
\sigma(x) \equiv (1 +d(O,x)^2 )^{1/2},
\]
where $d$ denotes the Riemannian distance function on
$\mathcal{S}$. The function $\sigma$ is used to define the following
weighted $L^2$ norm:
\begin{equation}
\label{Definition:WeightedSobolev}
\Vert u\Vert_\delta \equiv \left(\int_{\mathcal{S}} |u|^2 \sigma^{-2\delta-3} \mbox{d}x\right)^{1/2},
\end{equation}
for $\delta\in \mathbb{R}$. In particular, if $\delta=-3/2$ one
recovers the usual $L^2$ norm. Different choices of origin give rise
to equivalent weighted norms ---as mentioned before, the convention of
indices used in the definition of the norm \eqref{Definition:WeightedSobolev}
follows the one of  Bartnik \cite{Bar86}. The fall off conditions of the various fields will be expressed in
terms of weighted Sobolev spaces $H^s_\delta$ consisting of functions
for which the norm
\[
\Vert u\Vert_{s,\delta} \equiv \sum_{0\leq |\alpha| \leq s} \Vert
D^\alpha u\Vert_{\delta-|\alpha|} < \infty,
\]
with $s$ a non-negative integer, and where
$\alpha=(\alpha_1,\alpha_2,\alpha_3)$ is a multiindex,
$|\alpha|=\alpha_1 +\alpha_2 +\alpha_3$.  We say that $u\in
H^\infty_\delta$ if $u\in H^s_\delta$ for all $s$.  We will say that a
spinor or a tensor belongs to a function space if its norm does. For
instance, the notation $\zeta_{AB}\in H^s_\delta$ is a short hand
notation for
$(\zeta_{AB}\hat\zeta^{AB}+\zeta_A{}^A\hat\zeta_B{}^B)^{1/2}\in
H^s_\delta$.  

\medskip
We will make use of the following result: 

\begin{lemma}
Let $u\in H^\infty_\delta$. Then $u$ is smooth (i.e. $C^\infty$)  over
$\mathcal{S}$ and has a fall off at infinity such that $D^l u = o(r^{\delta-|l|})$.
\end{lemma}

The smoothness of $u$ follows from the Sobolev embedding theorems. The
proof of the behaviour at infinity of $u$ can be found in \cite{Bar86}
---cfr. Theorem 1.2 (iv)--- while the decay for the derivatives follows from
the definition of the weighted Sobolev norms.

\medskip
\noindent
\textbf{Remark.} Here $r$ is a radial coordinate on the asymptotic end ---see the next section for details.

\medskip
We also note the following multiplication lemma
---cfr. e.g. Theorem 5.6 in \cite{Can81}.

\begin{lemma}
\label{Lemma:Multiplication}
Let $u\in H^\infty_{\delta_1}$, $v\in H^\infty_{\delta_2}$. Then
\[
u v \in H^\infty_{\delta_1+\delta_2+\varepsilon}, \quad \varepsilon>0.
\]
\end{lemma}

\noindent
\textbf{Notation.}  We will often write $u=o_\infty(r^\delta)$ for
$u\in H^\infty_\delta$ at an asymptotic end. 

\medskip
For the present applications we will require a somehow finer
multiplication lemma concerning the behaviour at infinity. For this we
exploit the fact that we are working with smooth functions. More precisely:

\begin{lemma}
\label{Lemma:FinerMultiplication}
Let $u=o_\infty(r^{\delta_1})$, $v=o_\infty(r^{\delta_2})$ and $w=O(r^\gamma)$. Then
\[
uv =o(r^{\delta_1+\delta_2}), \quad uw=o(r^{\delta_1+\gamma}).
\]
\end{lemma}

\begin{proof}
Let $\partial \mathcal{S}_r$ denote the surfaces of constant $r$. For
sufficiently large $r$ (so that one is in an asymptotic end), the
surface $\partial \mathcal{S}_{r}$ has the topology of the
2-sphere. Now, the functions $u, \; v$ are continuous and the surfaces
$\partial \mathcal{S}_{r}$ are compact.  Therefore, for sufficiently
large $r$ the functions
\[
 f(r)\equiv \max_{\partial \mathcal{S}_r}|u r^{-\delta_1}|, \quad g(r)\equiv \max_{\partial \mathcal{S}_r}|v r^{-\delta_2}|,
\]
 are finite and well defined. Furthermore $r^{\delta_1}|u|\leq f(r)$, $r^{\delta_2}|v|\leq g(r)$.
By construction, one has that $f(r)=o(1)$
and $g(r)=o(1)$ ---that is, $f,\; g\rightarrow 0$ for $r\rightarrow
\infty$. One also has that $|w r^{-\gamma}|$ is bounded by a constant
$C$. Hence,
\begin{eqnarray*}
&& |uv| \leq f(r)g(r) r^{\delta_1+\delta_2} = o(r^{\delta_1+\delta_2}), \\
&& |uw| \leq f(r) r^{\delta_1} |w| \leq Cf(r) r^{\delta_1+\gamma} =o(r^{\delta_1+\gamma}), 
\end{eqnarray*}
 from where the desired result follows.
\end{proof}

\medskip
\noindent
\textbf{Remark.} The lemmas extend to symmetric spatial spinors with
even number of indices by the Cauchy-Schwartz inequality.

\subsection{Decay assumptions}
As mentioned before, our analysis will be restricted to initial data
sets $(\mathcal{S},h_{ab},K_{ab})$ with 2 asymptotic ends. Without
loss of generality one of the ends will be denoted by the
subscript/superscript $+$ on the relevant objects, while those of the
other end by $-$. Often, when no confusion arises the
subscript/superscript will be dropped.

\medskip
\noindent
\textbf{Remark.} We do not need to assume any topological restriction apart from paracompactness, orientability and the requirement of 2 asymptotically flat ends. Hence, we can have an arbitrary number of handles. For black holes, this means that we can handle Misner-type data with several black holes \cite{Mis63}.

\bigskip
The standard assumption for asymptotic flatness is that on each end it is possible to introduce asymptotically
Cartesian coordinates $x^i_\pm$ with $r=((x^1_\pm)^2 +
(x^2_\pm)^2 + (x^3_\pm)^2)^{1/2}$, such that the intrinsic metric
and extrinsic curvature of $\mathcal{S}$ satisfy
\begin{subequations}
\begin{eqnarray}
&& h_{ij} = -\delta_{ij}+o_\infty(r^{-1/2}), \label{decay1} \\
&& K_{ij} = o_\infty(r^{-3/2}) \label{decay2}.
\end{eqnarray}
\end{subequations}
Note that the decay conditions
\eqref{decay1} and \eqref{decay2} allow for data containing
non-vanishing linear and angular momentum. For the purposes of our analysis, it will be necessary to have a bit
more information about the behaviour of leading terms in $h_{ij}$ and
$K_{ij}$. More precisely, we will require the initial data to be
\emph{asymptotically Schwarzschildean} in some suitable sense. For
example, in 
\cite{BaeVal10a} the assumptions
\begin{subequations}
\begin{eqnarray}
&& h_{ij} = -\left(1+ 2m_\pm r^{-1}\right)\delta_{ij} + o_\infty(r^{-3/2}), \label{OldDecay1} \\
&& K_{ij} = o_\infty(r^{-5/2}), \label{OldDecay2}
\end{eqnarray}
\end{subequations}
have been used. This class of data
can be described as \emph{asymptotically non-boosted
Schwarz\-schildean}. Here, we consider a more general class of data which
includes boosted Schwarzschild data. Following \cite{BeiOMu87,Hua10} we assume
\begin{subequations}
\begin{eqnarray}
&& h_{ij} = -\left(1+\frac{2A_\pm}{r}\right)\delta_{ij} - \frac{\alpha_\pm}{r}\left( \frac{2x_ix_j}{r^2}-\delta_{ij}\right)+o_\infty(r^{-3/2}), \label{BoostedDecay1} \\
&& K_{ij} = \frac{\beta_\pm}{r^2}\left( \frac{2x_ix_j}{r^2}-\delta_{ij} \right) + o_\infty(r^{-5/2}),  \label{BoostedDecay2}
\end{eqnarray}
\end{subequations}
where $\alpha_\pm$ and $\beta_\pm$ are smooth functions on the
2-sphere and $A_\pm$ denotes a constant. The functions $\alpha$ and
$\beta$ are related to each other via the vacuum constraint equations
\eqref{Hamiltonian} and \eqref{Momentum}. We will later need to be
more specific about their particular form. The decay assumption for
the metric, equation \eqref{decay1} and hence also
\eqref{BoostedDecay1}, is included  in the analysis of
\cite{Can81}.

Important for our analysis is that boosted
Schwarzschild data is of this form ---see \cite{BeiOMu87}. It is noticed that
a second fundamental form of the type given by (\ref{BoostedDecay2})
is, in general, not trace-free:
\[
K_i{}^i = \frac{\beta_\pm}{r^2} + o_\infty(r^{-5/2}).
\]

Henceforth, we drop the superscripts/subscripts $\pm$ for ease of
presentation. If $\pm$ appears in any formula, $+$ is assumed for the
$(+)$-end, $-$ for the $(-)$-end. For the $\mp$ sign we assume the
opposite.

\subsection{ADM mass and momentum}
The ADM energy, $E$, and momentum, $p_i$, at each end are given by the integrals:
\begin{eqnarray*}
&& E= \frac{1}{16\pi} \int_{\partial\mathcal{S}_\infty} \delta^{ij} \left(\partial_i h_{jk}- \partial_k h_{ij} \right)\frac{x^k}{r} \mbox{d}S, \\
&& p_i = \frac{1}{8\pi} \int_{\partial\mathcal{S}_\infty} \left(K_{ij} - K h_{ij} \right)\frac{x^j}{r} \mbox{d}S,
\end{eqnarray*}
so that the ADM 4-momentum covector is given by $p_\mu=(E,p_i)$. 
In what follows it will be assumed that $p_\mu$ is timelike ---that is, $p_\mu p^\mu>0$. The need of this assumption will become clear in the sequel. From the ADM 4-momentum, we define the constants 
\[
m\equiv\sqrt{p^\nu p_\nu}, \quad p^2\equiv E^2-m^2.
\]

\subsection{Asymptotically Schwarzschildean data}
Boosted Schwarzschild data sets---i.e. initial data for the Schwarzschild
spacetime for which $p_i\neq 0$ satisfy the decay assumptions
\eqref{BoostedDecay1}-\eqref{BoostedDecay2}. This type of data satisfies:
\begin{eqnarray*}
&& A=\frac{m}{\sqrt{1-v^2}},\\
&& \alpha = 2m\left(1+ 2\frac{(n\cdot v)^2}{1-v^2}\right)\left(1+\frac{(n\cdot v)^2}{1-v^2} \right)^{-1/2}-\frac{2m}{\sqrt{1-v^2}}, \\
&& \beta = 2m \frac{n\cdot v}{1-v^2}\left(\frac{3}{2}+\frac{(n\cdot v)^2}{1-v^2}\right)\left(1+\frac{(n\cdot v)^2}{1-v^2}\right)^{-3/2},
\end{eqnarray*} 
where $n^i\equiv x^i/r$, $n\cdot v\equiv n^i v_i$,
$v^2\equiv\delta^{ij}v_i v_j$, $v_i$ is a constant 3-covector
---cfr. \cite{BeiOMu87}, and $m_\pm=m$. Note that if $v_i=0$ then
\eqref{BoostedDecay1}-\eqref{BoostedDecay2} reduce to
\eqref{OldDecay1}-\eqref{OldDecay2}. It can be checked that
\[
E= \frac{m}{\sqrt{1-v^2}}, \quad p_i=\frac{mv_i}{\sqrt{1-v^2}}.
\]
Rewriting this in terms of $(E,p_i)$, we get
\begin{equation}\label{alphabeta}
A=E,\qquad
\alpha = \frac{2m^2+4(n\cdot p)^2}{\sqrt{m^2+(n\cdot p)^2}}-2E, \qquad
\beta = \frac{(n\cdot p) E (3m^2+2(n\cdot p)^2)}{(m^2+(n\cdot p)^2)^{3/2}},
\end{equation} 
where $n\cdot p=n^i p_i=r^{-1}x^ip_i$.

\medskip
\noindent
\textbf{Assumption.}
In the sequel, we will restrict our analysis to initial data sets
which are asymptotically Schwarzschildean to the order given by
\eqref{BoostedDecay1}-\eqref{BoostedDecay2}. For any asymptotically
flat data that admits ADM 4-momentum, one can compute $(E,p_i)$, and
then try to find coordinates that cast the metric and extrinsic
curvature into the form \eqref{BoostedDecay1}-\eqref{BoostedDecay2}
with $(A,\alpha,\beta)$ given by \eqref{alphabeta} with $m=m_\pm$. If
this is possible, we will say that the data is \emph{asymptotically
Schwarzschildean}. We expect this to be the case for a large class of
data. The initial data sets excluded by this assumption will be deemed
pathological. Examples of such pathological cases can be found in
\cite{Hua10}. We stress that all data of the form
\eqref{OldDecay1}-\eqref{OldDecay2} is included in our more general
analysis.

\medskip
The need to restrict our analysis to asymptotically Schwarzschildean
data as defined in the previous paragraph will become evident in the
sequel, where we need to find an asymptotic solution to the spatial Killing spinor equation.

\section{Asymptotic behaviour of the spatial Killing spinors}
\label{Section:AB}

In this section we discuss in some detail the asymptotic behaviour of
solutions to the spatial Killing spinor equation on an asymptotically
Euclidean manifold. We begin by studying the asymptotic behaviour of
the appropriate Killing spinor in the Kerr spacetime. Then, we will
impose the same asymptotics on the approximate Killing spinor on a
slice of a much more general spacetime. In what follows, we
concentrate our discussion on a particular asymptotic end.

\subsection{Asymptotic form of the stationary Killing vector}
As seen in section \ref{Section:Basics}, the Killing spinor of the
Kerr spacetime gives rise to its stationary Killing vector
$\xi^\mu$. It will be assumed that the spacetime is such that
$p_\mu=(E, p_i)$ is timelike at each asymptotic end. If this is
the case, then $p^\mu/\sqrt{p^\nu p_\nu}$ gives the asymptotic
direction of the stationary Killing vector at each end ---see
e.g.\cite{BeiChr96}. Let
\[
 m\equiv\sqrt{p^\nu p_\nu}, \quad p^2\equiv E^2-m^2.
\]
Recall now, that $\xi$ and $\xi_{AB}$ denote the lapse and shift of
the spinorial counterpart, $\xi^{AA'}$, of the Killing vector
$\xi^\mu$. One finds that for non-boosted initial data sets of the form
\eqref{OldDecay1}-\eqref{OldDecay2}, one has in terms of the asymptotic Cartesian coordinates and spin frame, that
\[
\xi= \pm\sqrt{2} +o_\infty(r^{-1/2}), \quad \xi_{\mathbf{AB}}= o_\infty(r^{-1/2}).
\]
The factor of $\sqrt{2}$ arises due to the particular
normalisations used in the space spinor formalism. This particular
form of the asymptotic behaviour of the Killing vector has been
discussed in \cite{BaeVal10a}.

\medskip
Now consider the more general case given by
\eqref{BoostedDecay1}-\eqref{BoostedDecay2}. Again, adopting
asymptotically Cartesian coordinates, we extend $p_i$ to a constant
covector field on the asymptotic end. In terms of the associated
asymptotically Cartesian spin frame, we then define $p_{\mathbf{AB}}
\equiv \sigma^i{}_{\mathbf{AB}} p_i$.  One finds that
\begin{equation}
\xi= \pm\frac{\sqrt{2}E}{m}+o_\infty(r^{-1/2}), \quad \xi_{\mathbf{AB}}=\pm\frac{\sqrt{2}p_{\mathbf{AB}}}{m}+o_\infty(r^{-1/2}).\label{xi_first_leading}
\end{equation}

We see that the conditions \eqref{xi_first_leading} are well defined even if we do not have a Killing vector in the spacetime.
Hence, for the general case when the metric satisfies
\eqref{BoostedDecay1}-\eqref{BoostedDecay2} and the ADM 4-momentum is well defined, we can still impose the asymptotics \eqref{xi_first_leading} for our approximate Killing spinor. We will however need to assume that the functions in the metric are given by \eqref{alphabeta}. Otherwise we will not be able to assume $\xi_{ABCD}\in H^\infty_{-3/2}$, as we will do in the next section. We will later see that this condition is important for the solvability of the elliptic equation \eqref{ApproximateKillingEquation}.

\subsection{Asymptotic form of the spatial Killing spinor}
\label{Section:DecayKappa}
In the sequel, given an initial data set $(\mathcal{S},h_{ab},K_{ab})$
satisfying the decay conditions
\eqref{BoostedDecay1}-\eqref{BoostedDecay2} with $A$, $\alpha$ and
$\beta$ given by \eqref{alphabeta} with $m=m_\pm$, it will be
necessary to show that it is always possible to solve the equation
\begin{equation}
\label{AsymptoticSpatialKillingSpinorEquation}
\nabla_{(AB}\kappa_{CD)}=o_\infty(r^{-3/2}),
\end{equation}
order by order without making any further assumptions on the
data. A direct calculation allows us to verify that:

\begin{lemma}
Let $(\mathcal{S},h_{ab},K_{ab})$ denote an initial data set for the
vacuum Einstein field equations satisfying at each asymptotic end the
decay conditions \eqref{BoostedDecay1}-\eqref{BoostedDecay2} with $A$,
$\alpha$ and $\beta$ given by \eqref{alphabeta} and $m$ the ADM mass
of the respective end. Then
\begin{eqnarray}
&& \kappa_{\mathbf{AB}} = 
\mp\frac{\sqrt{2}E}{3m}\left (1+\frac{2E}{r}\right)x_{\mathbf{AB}} \nonumber \\
&&\hspace{2cm}\pm\frac{2\sqrt{2}}{3m}\left(1
+\frac{4E}{r}
-\frac{m^2+2(n\cdot p)^2}
{\sqrt{m^2+(n\cdot p)^2}r}
\right)p_{\mathbf{Q}(\mathbf{A}}x_{\mathbf{B})}{}^{\mathbf{Q}}
+o_\infty(r^{-1/2}), \label{KillingSpinorLeading}
\end{eqnarray}
with $x_{\mathbf{AB}}$ as in \eqref{CartesianSpinor}, and $n\cdot
p=r^{-1}x^{\mathbf{AB}} p_{\mathbf{AB}}$ satisfies equation \eqref{AsymptoticSpatialKillingSpinorEquation}.
\end{lemma}

\medskip
\noindent
\textbf{Remark.} Formula \eqref{KillingSpinorLeading} implies the
following expansions for $\xi$ and $\xi_{\mathbf{AB}}$:
\begin{subequations}
\begin{eqnarray}
&&\xi=\pm\frac{\sqrt{2}E}{m}\mp\frac{\sqrt{2}E(m^2+2(n\cdot p)^2)}{m\sqrt{m^2+(n\cdot p)^2}}r^{-1}+o_\infty(r^{-3/2}),\label{xi0Leading}\\
&& \xi_{\mathbf{AB}}=
\pm\Biggl(-\frac{2\sqrt{2}E}{m}
+\frac{\sqrt{2}(E^2+4(n\cdot p)^2)}{m\sqrt{m^2+(n\cdot p)^2}}
+\frac{mE^2}{\sqrt{2}(m^2+(n\cdot p)^2)^{3/2}}
\Biggr)(n\cdot p)r^{-2}x_{\mathbf{AB}}\nonumber\\
&&\hspace{2cm}\pm\Biggl(\frac{\sqrt{2}}{m}
+\frac{2\sqrt{2}E}{mr}
-\frac{2\sqrt{2}(m^2+2(n\cdot p)^2)}{m\sqrt{m^2+(n\cdot p)^2}r}
\Biggr)p_{\mathbf{AB}}
+o_\infty(r^{-3/2}).\label{xi2Leading}
\end{eqnarray}
\end{subequations}
In the case of non-boosted data the expansions \eqref{KillingSpinorLeading}, \eqref{xi0Leading} and \eqref{xi2Leading} reduce to
\begin{eqnarray*}
&& \kappa_{\mathbf{AB}} = 
\mp\frac{\sqrt{2}}{3}\left (1+\frac{2m}{r}\right)x_{\mathbf{AB}}
+o_\infty(r^{-1/2}),\\
&& \xi=\pm\sqrt{2}\mp\sqrt{2}mr^{-1}+o_\infty(r^{-3/2}),\\
&& \xi_{\mathbf{AB}}=o_\infty(r^{-3/2}),
\end{eqnarray*}
as discussed in \cite{BaeVal10a}. 

\subsection{Existence and uniqueness of spinors with Killing spinor asymptotics}
In this section we prove that given a spinor $\kappa_{AB}$ satisfying
equation \eqref{xi_first_leading} and
\eqref{AsymptoticSpatialKillingSpinorEquation}, then the asymptotic
expansion \eqref{KillingSpinorLeading} is unique up to a translation.

\begin{theorem}\label{ExistensKillingSpinorAsymptotics}
Assume that on an asymptotic end of the slice $\mathcal{S}$, one has
an asymptotically Cartesian coordinate system such that
\eqref{BoostedDecay1}-\eqref{BoostedDecay2} hold. Then there exists
\begin{equation}
\kappa_{\mathbf{AB}}=o_\infty(r^{3/2}),\label{AsymptoticAssumptions1}
\end{equation}
such that
\begin{equation}
\xi_{\mathbf{ABCD}}=o_\infty(r^{-3/2}), \quad 
\xi_{\mathbf{AB}}=\pm\frac{\sqrt{2}p_{\mathbf{AB}}}{m}+o_\infty(r^{-1/2}), \quad 
\xi=\pm\frac{\sqrt{2}E}{m}+o_\infty(r^{-1/2}).
\label{AsymptoticAssumptions2}
\end{equation}
The spinor 
$\kappa_{\mathbf{AB}}$ is unique up to order $o_\infty(r^{-1/2})$,
apart from a (complex) constant term.
\end{theorem}

\medskip
\noindent
\textbf{Remark 1.} The complex constant term arising in Theorem \ref{ExistensKillingSpinorAsymptotics}
contains 6 real parameters. In the sequel, given a particular choice
of asymptotically Cartesian coordinates and frame, we will set this
constant term to zero. Note that a change of asymptotically
Cartesian coordinates would introduce a similar term containing only 3
real parameters ---which by construction could be removed by a
suitable choice of gauge. In what follows, we will use coordinate
independent expressions, and therefore, this translational ambiguity will not affect the result.

\medskip
\noindent
\textbf{Remark 2.}
Note that $\xi_{\mathbf{ABCD}}=o_\infty(r^{-3/2})$
implies $\xi_{ABCD}\in L^2$. The conditions in Theorem
\ref{ExistensKillingSpinorAsymptotics} are
coordinate independent. 

\begin{proof}
A direct calculation shows that the expansion
\eqref{KillingSpinorLeading} yields \eqref{xi0Leading},
\eqref{xi2Leading} and
$\xi_{\mathbf{ABCD}}=o_\infty(r^{-3/2})$. Hence,
\eqref{KillingSpinorLeading} gives a solution of the desired form. In
order to prove uniqueness we make use of the linearity of the
integrability conditions \eqref{Dxi1}-\eqref{Dxi3} and
\eqref{nabla2xi0a}-\eqref{nabla2xi0c}. Note that the translational
freedom gives an ambiguity of a constant term in
$\kappa_{\mathbf{AB}}$. Let
\begin{equation}
\mathring\kappa_{\mathbf{AB}} \equiv 
\mp\frac{\sqrt{2}E}{3m}\left (1+\frac{2E}{r}\right)x_{\mathbf{AB}}
\pm\frac{2\sqrt{2}}{3m}\left(1
+\frac{4E}{r}
-\frac{m^2+2(n\cdot p)^2}
{\sqrt{m^2+(n\cdot p)^2}r}
\right)p_{\mathbf{Q}(\mathbf{A}}x_{\mathbf{B})}{}^{\mathbf{Q}}.
\label{KillingSpinorLeadingTerms}
\end{equation}
Let $\breve\kappa_{\mathbf{AB}}$, be an arbitrary solution to the system \eqref{xi_first_leading}, \eqref{AsymptoticSpatialKillingSpinorEquation}. Furthermore, let
$\kappa_{\mathbf{AB}}\equiv\breve\kappa_{\mathbf{AB}}-\mathring\kappa_{\mathbf{AB}}$. We
then have
\[
\xi_{ABCD}=o_\infty(r^{-3/2}), \quad  \xi_{AB}=o_\infty(r^{-1/2}), \quad  \xi=o_\infty(r^{-1/2}), \quad  \kappa_{AB}=o_\infty(r^{3/2}).
\]
 To obtain the desired conclusion we only need to prove that
 $\kappa_{\mathbf{AB}}=C_{\mathbf{AB}}+o_\infty(r^{-1/2})$, where
 $C_{\mathbf{AB}}$ is a constant. This is equivalent to
 $D_{AB}\kappa_{CD}=o_\infty(r^{-3/2})$. Note that we now have
 coordinate independent statements to prove.

\medskip
We note that from \eqref{BoostedDecay1}-\eqref{BoostedDecay2} it follows that
\[
K_{ABCD}=o_\infty(r^{-2+\varepsilon}), \quad \Psi_{ABCD}=o_\infty(r^{-3+\varepsilon}),
\]
with $\varepsilon>0$. From \eqref{SenDiffKappaSplit} and
Lemma \ref{Lemma:Multiplication} we have
\begin{eqnarray*}
&& D_{AB}\kappa_{CD}=\xi_{ABCD}-\tfrac{1}{3}\epsilon_{A(C}\xi_{D)B}-\tfrac{1}{3}\epsilon_{B(C}\xi_{D)A}-\tfrac{1}{3}\epsilon_{A(C}\epsilon_{D)B}\xi-K_{AB(C}{}^E\kappa_{D)E} \\
&& \phantom{ D_{AB}\kappa_{CD}}=o_\infty(r^{-1/2+\varepsilon}).
\end{eqnarray*}
Integrating the latter yields 
\[
\kappa_{AB}=o_\infty(r^{1/2+\varepsilon}).
\]
The constant of integration is incorporated in the remainder
term. Repeating this procedure allows to gain an $\varepsilon$ in the
decay so that
\[
D_{AB}\kappa_{CD}=o_\infty(r^{-1/2}), \quad \kappa_{AB}=o_\infty(r^{1/2}).
\]
Estimating all terms in \eqref{nabla2xi0a}, \eqref{nabla2xi0b} and
\eqref{nabla2xi0c} gives
\begin{subequations}
\begin{align}
\nabla^{AB}\nabla_{AB}\xi={}&
\xi_{AB}\nabla^{AB}K 
+o_\infty(r^{-7/2}) \nonumber \\
={}&o_\infty(r^{-7/2+\varepsilon}),\label{estnabla2xi0a}\\
\nabla^C{}_{(A}\nabla_{B)C}\xi
={}&\tfrac{1}{2}\Omega_{ABCD}\nabla^{CD}\xi 
-\tfrac{1}{3}K \nabla_{AB}\xi \nonumber \\ 
={}&o_\infty(r^{-7/2+\varepsilon}),\label{estnabla2xi0b}\\
\nabla_{(AB}\nabla_{CD)}\xi
={}&
+\tfrac{1}{2} \hat\Psi_{ABCD}\xi 
-\tfrac{5}{2} \Psi_{ABCD}\xi 
-\tfrac{2}{3} \hat\Psi_{(ABC}{}^E \xi_{D)E} 
-\tfrac{10}{3} \Psi_{(ABC}{}^E \xi_{D)E}\nonumber\\ 
&+\Omega_{ABCDEL} \xi^{EL} 
+\tfrac{2}{5} \xi_{(CD} \nabla _{AB)}K
-3 \Omega_{E(BCD} \nabla_{A)}{}^E\xi\nonumber\\ 
&-3\kappa^{EL}\nabla_{L(D}\Psi_{ABC)E}
+3\kappa_{(A}{}^E \nabla_D{}^L\Psi_{BC)EL}
+o_\infty(r^{-7/2}) \nonumber \\
={}& o_\infty(r^{-7/2+\varepsilon}). \label{estnabla2xi0c}
\end{align}
\end{subequations}
Hence, $\nabla_{AB} \nabla_{CD}\xi=o_\infty(r^{-7/2+\varepsilon})$,
and therefore $D_{AB}
D_{CD}\xi=o_\infty(r^{-7/2+\varepsilon})$. Integrating this yields
$D_{AB}\xi=o_\infty(r^{-5/2+\varepsilon})$. In this step the constants
of integration are forced to vanish by the condition
$D_{AB}\xi=o_\infty(r^{-3/2})$, which is a consequence of
$\xi=o_\infty(r^{-1/2})$. Integrating
$D_{AB}\xi=o_\infty(r^{-5/2+\varepsilon})$ and using
$\xi=o_\infty(r^{-1/2})$ to remove the constants of integration 
yields 
\[
\xi=o_\infty(r^{-3/2+\varepsilon}).
\]
Estimating all terms in \eqref{Dxi1}, \eqref{Dxi2} and \eqref{Dxi3} yields
\begin{subequations}
\begin{eqnarray}
&& \nabla^{AB}\xi_{AB}
= o_\infty(r^{-7/2+\varepsilon}),\label{estDxi1}\\
&& \nabla^C{}_{(A}\xi_{B)C}
=\tfrac{3}{2}\Psi_{ABCD}\kappa^{CD} 
-\tfrac{2}{3}K\xi_{AB} 
-\tfrac{1}{2}\Omega_{ABCD}\xi^{CD} 
+\nabla_{AB}\xi 
+o_\infty(r^{-5/2}) \nonumber \\
&& \phantom{\nabla^C{}_{(A}\xi_{B)C}}= o_\infty(r^{-5/2+\varepsilon}),\label{estDxi2}\\
&& \nabla_{(AB}\xi_{CD)}
=3\Psi_{E(ABC}\kappa_{D)}{}^E 
-\Omega_{E(ABC}\xi_{D)}{}^E 
+o_\infty(r^{-5/2}) \nonumber \\
&& \phantom{\nabla_{(AB}\xi_{CD)}}=o_\infty(r^{-5/2+\varepsilon}).\label{estDxi3}
\end{eqnarray}
\end{subequations}
Hence, $\nabla_{AB} \xi_{CD}=o_\infty(r^{-5/2+\varepsilon})$, and therefore $D_{AB} \xi_{CD}=o_\infty(r^{-5/2+\varepsilon})$.
Integrating and using $\xi_{AB}=o_\infty(r^{-1/2})$ to remove the constants of integration yields 
\[
\xi_{AB}=o_\infty(r^{-3/2+\varepsilon}).
\]
Now,
\begin{eqnarray*}
&& D_{AB}\kappa_{CD}=\xi_{ABCD}-\tfrac{1}{3}\epsilon_{A(C}\xi_{D)B}-\tfrac{1}{3}\epsilon_{B(C}\xi_{D)A}-\tfrac{1}{3}\epsilon_{A(C}\epsilon_{D)B}\xi-K_{AB(C}{}^E\kappa_{D)E} \nonumber \\
&& \phantom{D_{AB}\kappa_{CD}} =o_\infty(r^{-3/2+\varepsilon}).
\end{eqnarray*}
Integrating the latter we get 
\[
\kappa_{AB}=C_{AB}+o_\infty(r^{-1/2+\varepsilon}),
\]
 where $C_{AB}$ is a constant in some frame. To get a frame independent statement one can still use the estimate
 $\kappa_{AB}=o_\infty(r^{\varepsilon})$. Reevaluating the estimates
 \eqref{estnabla2xi0a}, \eqref{estnabla2xi0b} and
 \eqref{estnabla2xi0c} yields
\begin{eqnarray*}
&& \nabla^{AB}\nabla_{AB}\xi=
o_\infty(r^{-7/2}),\\
&& \nabla^C{}_{(A}\nabla_{B)C}\xi
=o_\infty(r^{-9/2+\varepsilon}),\\
&& \nabla_{(AB}\nabla_{CD)}\xi
=o_\infty(r^{-7/2}).
\end{eqnarray*}
Hence, one obtains
\[
\nabla_{AB}\nabla_{CD}\xi=o_\infty(r^{-7/2}).
\]
Integrating as before, we get
\[
\xi=o_\infty(r^{-3/2}).
\]
Finally, we can reevaluate the estimates \eqref{estDxi2} and \eqref{estDxi3}, to get
\begin{eqnarray*}
&& \nabla^C{}_{(A}\xi_{B)C}
=o_\infty(r^{-5/2}),\\
&& \nabla_{(AB}\xi_{CD)}
= o_\infty(r^{-5/2}).
\end{eqnarray*}
Combining this with \eqref{estDxi1}, we obtain
\[
\nabla_{AB}\xi_{CD}=o_\infty(r^{-5/2}).
\]
 Integrating as before, we get 
\[
\xi_{AB}=o_\infty(r^{-3/2}).
\]
Hence,
\[
D_{AB}\kappa_{CD}=\xi_{ABCD}-\tfrac{1}{3}\epsilon_{A(C}\xi_{D)B}-\tfrac{1}{3}\epsilon_{B(C}\xi_{D)A}-\tfrac{1}{3}\epsilon_{A(C}\epsilon_{D)B}\xi-K_{AB(C}{}^E\kappa_{D)E}=o_\infty(r^{-3/2}),
\]
 from where the result follows.
\end{proof}

 From the asymptotic solutions we can obtain a globally defined spinor $\mathring{\kappa}_{AB}$ on $\mathcal{S}$ that will act as a seed for our approximate Killing spinor.

\begin{corollary}\label{corkapparing}
There are spinors $\mathring{\kappa}_{AB}$, defined everywhere on
$\mathcal{S}$, such that the asymptotics at each end is given by
\eqref{KillingSpinorLeading}, where opposite signs are used at each
end. Different choices of $\mathring{\kappa}_{AB}$ can only differ by
a spinor in $H^\infty_{-1/2}$.
\end{corollary}
\medskip
\noindent
\textbf{Remark.}  The opposite signs at each end are motivated by
looking at the explicit example of standard Kerr data.
\begin{proof}
Theorem \ref{ExistensKillingSpinorAsymptotics} gives the existence at
each end. Smoothly cut off these functions, and paste them
together. This gives a smooth spinor $\mathring{\kappa}_{AB}$ defined
everywhere on $\mathcal{S}$. Furthermore
$\nabla_{(AB}\mathring{\kappa}_{CD)}\in H^\infty_{-3/2}$.
\end{proof}

\section{The approximate Killing spinor equation in asymptotically
  Euclidean manifolds}
\label{Section:ApproximateKSinAEM}

In this section we study the invertibility properties of the
approximate Killing spinor operator $L:\mathfrak{S}_2 \rightarrow
\mathfrak{S}_2$ given by equation \eqref{ApproximateKillingEquation}
on a manifold $\mathcal{S}$ which is asymptotically Euclidean in the
sense discussed in section
\ref{Section:AsymptoticallyEuclideanData}. In order to do so, we first
present some adaptations to our context of results for elliptic
equations that can be found in \cite{Can81,ChrOMu81,Loc81}.

\subsection{Ancillary results of the theory of elliptic equations on asymptotically Euclidean manifolds}

\subsubsection{Asymptotic homogeneity of $L$}

Let $u$ be the vector given by equation \eqref{vectorU}. The elliptic
operator defined by \eqref{ApproximateKillingEquation} can be written matricially in the
form
\[
(A^{ij}+a^{ij}_2)D_i D_j u + a^i_1 D_i u + a_0u =0,
\]
where $A^{ij}$ corresponds to the matrix associated to the
elliptic operator with constant coefficients $L'$ given by equation
\eqref{matrixA}, and $a^{ij}_2$, $a^j_1$, $a_0$ are matrix valued
functions such that 
\[
a^{ij}_2 \in H^\infty_{-1/2}, \quad  a^j_1 \in H^\infty_{-3/2}, \quad a_0 \in H^\infty_{-5/2}.
\]
Using the terminology of \cite{Can81,Loc81} we say that $L$ is an
\emph{asymptotically homogeneous elliptic operator}\footnote{The sharp
conditions for a second order elliptic operator to be asymptotically homogeneous
are that
\[
a_2^{ij} \in H^\infty_\delta, \quad a_1^i \in H^\infty_{\delta-1},
\quad a_0 \in H^\infty_{\delta-2}, 
\]
for $\delta <0$. As one sees, our operator $L$ satisfies these
conditions with a margin.}. This is the
standard assumption on elliptic operators on asymptotically Euclidean
manifolds. It follows from \cite{Can81}, Theorem 6.3 that:

\begin{theorem}
The elliptic operator 
\[
L: H^{2}_{\delta} \rightarrow H^{0}_{\delta-2}, 
\]
with $\delta$ is not a non-negative integer is a linear bounded operator with
finite dimensional Kernel and closed range.
\end{theorem}

\subsubsection{The Kernel of $L$}
We investigate some relevant properties of the Kernel of $L$. This, in turn,
requires an analysis of the Kernel of the operator of the Killing
spinor equation \eqref{SpatialKillingSpinorEquation}.

\medskip
The following is an adaptation to the smooth spinorial setting of an
ancillary result from \cite{ChrOMu81}\footnote{The hypotheses in
  \cite{ChrOMu81} are much weaker than the ones presented here. The
  adaptation to the smooth setting has been chosen for simplicity.}.

\begin{theorem}
\label{Lemma:ChrOMu81}
Let $\nu_{A_1B_1\cdots A_pB_p}$ be a $C^\infty$ spinorial field over $\mathcal{S}$ such that 
\[
\nabla_{E_{m+1}F_{m+1}}\cdots \nabla_{E_1 F_1} \nu_{A_1B_1 \cdots A_pB_p} = H_{E_{m+1}F_{m+1}\cdots E_1 F_1 A_1 B_1 \cdots A_p B_p}
\]
with $m,\;p$ non-negative integers, and where $H_{E_{m+1}F_{m+1}\cdots
E_1 F_1 A_1 B_1 \cdots A_p B_p}$ is a linear combination of
$\nu_{A_1B_1 \cdots A_pB_p}$, $\nabla_{E_1F_1}\nu_{A_1B_1 \cdots
A_pB_p}$, $\ldots$, $\nabla_{E_mF_m}\cdots \nabla_{E_1F_1}\nu_{A_1B_1
\cdots A_pB_p}$ with coefficients $b_k$ where $k$ denotes the order of
the derivative the coefficient is associated to. If $b_k\in
H^\infty_{\delta_k}$ with
\[
k-m-1 > \delta_k, \quad 0\leq k \leq m
\]
and $\nu_{A_1B_1\cdots A_pB_p}\in H^\infty_\beta$, $\beta<0$, then 
\[
\nu_{A_1B_1\cdots A_pB_p}=0 \quad { on } \;\;\mathcal{S}.
\] 
\end{theorem}

This last result, together with Lemma \ref{Lemma:ThirdDerivative}
allows to show that there are no non-trivial Killing spinor candidates
that go to zero at infinity ---in \cite{ChrOMu81} an analogous result
has been proved for Killing vectors. More precisely,

\begin{proposition}
\label{Proposition:TrivialityKernel}
Let $\nu_{AB}\in H^\infty_{-1/2}$ such that
$\nabla_{(AB}\nu_{CD)}=0$. Then $\nu_{AB}=0$ on $\mathcal{S}$.
\end{proposition}

\begin{proof}
 From Lemma \ref{Lemma:ThirdDerivative} it follows that
$\nabla_{AB} \nabla_{CD} \nabla_{EF} \nu_{GH}$ can be expressed as a
linear combination of lower order derivatives with smooth coefficients
with the proper decay. Thus, Theorem \ref{Lemma:ChrOMu81} applies with $m=2$
and one obtains the desired result.
\end{proof}

\bigskip
We are now in the position to discuss the Kernel of the approximate
Killing spinor operator in the case of spinor fields that go to zero
at infinity. The following is the main result of this section.

\begin{proposition}\label{EllipticKernel}
Let $\nu_{AB}\in H^\infty_{-1/2}$. If $L(\nu_{AB})=0$, then $\nu_{AB}=0$. 
\end{proposition}

\begin{proof}
Using the identity \eqref{IntegrationbyParts} with
$\zeta_{ABCD}= \nabla_{(AB} \nu_{CD)}$ and assuming that
$L(\nu_{CD})=0$, one obtains
\[
\int_{\mathcal{S}} \nabla^{AB}\nu^{CD} \widehat{\nabla_{(AB}
\nu_{CD)}} \mbox{d}\mu = \int_{\partial\mathcal{S}_\infty}
n^{AB}\nu^{CD} \widehat{\nabla_{(AB}\nu_{CD)}} \mbox{d}S,
\]
where $\partial S_\infty$ denotes the sphere at infinity. Assume now,
that $\nu_{AB}\in H^\infty_{-1/2}$. It follows that $\nabla_{(AB} \nu_{CD)}
\in H^\infty_{-3/2}$ and furthermore, using Lemma
\ref{Lemma:FinerMultiplication} that 
\[
n^{AB} \nu^{CD}
\widehat{\nabla_{(AB}\nu_{CD)}} = o(r^{-2}).
\]
 The integration of the
latter over a finite sphere of sufficiently large radius is of type
$o(1)$. Thus one has that
\[
\int_{\partial\mathcal{S}_\infty}
n^{AB}\nu^{CD} \widehat{\nabla_{(AB}\nu_{CD)}} \mbox{d}S=0,
\]
 from where
\[
\int_{\mathcal{S}} \nabla^{AB}\nu^{CD} \widehat{\nabla_{(AB} \nu_{CD)}} \mbox{d}\mu =0.
\]
Therefore, one concludes that 
\[
\nabla_{(AB} \nu_{CD)}=0.
\]
That is, $\nu_{AB}$ has to be a spatial Killing spinor. Using
Proposition \ref{Proposition:TrivialityKernel} it follows that
$\nu_{AB}= 0$ on $\mathcal{S}$. 
\end{proof}

\subsubsection{The Fredholm alternative and elliptic regularity}

We will make use of the following adaptation of the Fredholm alternative
for second order asymptotically homogeneous elliptic operators
on asymptotically Euclidean manifolds ---cfr. \cite{Can81}.

\begin{theorem}
\label{Theorem:FredholmAlternative}
Let $A$ be an asymptotically homogeneous elliptic
operator of order 2 with smooth coefficients. Given $\delta<0$, the equation
\[
A(\zeta_{AB}) = f_{AB}, \quad f_{AB}\in H^0_{\delta-2},
\]
has a solution $\zeta_{AB}\in H^2_{\delta}$  if
\[
\int_{\mathcal{S}} f_{AB} \hat{\nu}^{AB} \mbox{d}\mu=0
\]
for all $\nu_{AB}$ satisfying  
\[
\nu_{AB} \in H^0_{-1-\delta}, \quad A^*(\nu_{AB})=0,
\]
where $A^*$ denotes the formal adjoint of $A$. 
\end{theorem}

In order to assert the regularity of solutions, we will need the
following elliptic estimate ---see expression (62) in the proof of
Theorem 6.3 of \cite{Can81}.

\begin{theorem}
\label{Lemma:Regularity}
Let $A$ be an asymptotically homogeneous elliptic operator of order 2
with smooth coefficients. Then for any $\delta\in \Real$ and any $s\geq
2$ there exists a constant $C$ such that for every $\zeta_{AB} \in
H^s_{loc} \cap H^0_\delta$, the following inequality holds
\[
\Vert \zeta_{AB} \Vert_{H^s_\delta} \leq C \left(  \Vert A(\zeta_{AB})
  \Vert_{H^{s-2}_{\delta-2}} + \Vert \zeta_{AB} \Vert_{H^{s-2}_\delta} \right).
\] 
\end{theorem}

\medskip
\noindent
\textbf{Notation.} $H^s_{loc}$ denotes the local Sobolev
space. That is, $u \in H^s_{loc}$ if for an arbitrary smooth function
$v$ with compact support, $uv \in H^s$.

\medskip
\noindent
\textbf{Remark.} If $A$ has smooth coefficients, and 
$A(\zeta_{AB})=0$ then it follows that all the
$H^s_\delta$ norms of $\zeta_{AB}$ are bounded by the $H^0_\delta$
norm. Thus, it follows that if a solution to $ A(\zeta_{AB})=0$
exists, it must be smooth ---\emph{elliptic regularity}.

\subsection{Existence of approximate Killing spinors}
We are now in the position of providing an existence proof to
solutions to equation \eqref{ApproximateKillingEquation} with the
asymptotic behaviour discussed in section \ref{Section:DecayKappa}.

\begin{theorem}
\label{Theorem:ExistenceKS}
Given an asymptotically Euclidean initial data set
$(\mathcal{S},h_{ab},K_{ab})$ satisfying the asymptotic conditions
\eqref{BoostedDecay1}-\eqref{BoostedDecay2} and \eqref{alphabeta},
there exists a smooth unique solution to equation
\eqref{ApproximateKillingEquation} with asymptotic behaviour at each
end given by \eqref{KillingSpinorLeading}.
\end{theorem}

\begin{proof} 
We consider the Ansatz
\[
\kappa_{AB} = \mathring{\kappa}_{AB} + \theta_{AB}, \quad \theta_{AB} \in H^2_{-1/2},
\]
with $\mathring{\kappa}$ given by Corollary \ref{corkapparing}. Substitution into equation \eqref{ApproximateKillingEquation} renders the following equation for the spinor
$\theta_{AB}$:
\begin{equation}
\label{elliptic:general}
L(\theta_{CD}) = -L(\mathring{\kappa}_{CD}).
\end{equation}
By construction it follows that 
\[
\nabla_{(AB} \mathring{\kappa}_{CD)}\in
H^\infty_{-3/2},
\]
 so that 
\[
F_{CD}\equiv-L(\mathring{\kappa}_{CD})\in H^\infty_{-5/2}.
\]
Using Theorem \ref{Theorem:FredholmAlternative} 
with $\delta=-1/2$, one concludes
that equation \eqref{elliptic:general} has a unique solution if
$F_{AB}$ is orthogonal to all $\nu_{AB}\in H^0_{-1/2}$ in the Kernel
of $L^*=L$. Proposition \ref{Proposition:TrivialityKernel} states that this Kernel is trivial. Thus, there are no restrictions on $F_{AB}$ and equation \eqref{elliptic:general} has a unique solution as
desired. Due to elliptic regularity, any $H^2_{-1/2}$
solution to the previous equation is in fact a $H^\infty_{-1/2}$
solution ---cfr. Lemma \ref{Lemma:Regularity}. Thus, $\theta_{AB}$ is
smooth. 
To see that $\kappa_{AB}$ does not depend on the particular choice of $\mathring\kappa_{AB}$, let $\mathring\kappa'_{AB}$, be another choice. Let $\kappa'_{AB}$ be the corresponding solution to \eqref{elliptic:general}. Due to Corollary \ref{corkapparing}, we have $\mathring\kappa_{AB}-\mathring\kappa'_{AB} \in H^\infty_{-1/2}$. 
Hence, we have $\kappa_{AB}-\kappa'_{AB}\in H^\infty_{-1/2}$ and 
$L(\kappa_{AB}-\kappa'_{AB})=0$. According to Proposition \ref{EllipticKernel}, $\kappa_{AB}-\kappa'_{AB}=0$, and the proof is complete. 
\end{proof}

\medskip
The following is a direct consequence of Theorem \ref{Theorem:ExistenceKS}, and will be
crucial for obtaining an invariant characterisation of Kerr data.

\begin{corollary}
\label{Corollary:Boundedness}
A solution, $\kappa_{AB}$, to equation
\eqref{ApproximateKillingEquation} with asymptotic behaviour given by \eqref{KillingSpinorLeading} satisfies $J<\infty$ where $J$ is the functional given by equation \eqref{functional}.
\end{corollary}

\begin{proof}
The functional $J$ given by equation
\eqref{functional} is the $L^2$ norm of
$\nabla_{(AB}\kappa_{CD)}$. Now, if $\kappa_{AB}$ is the solution
given by Theorem \ref{Theorem:ExistenceKS}, one has that
$\nabla_{(AB}\kappa_{CD)}\in H^0_{-3/2}$. In Bartnik's
conventions one has that 
\[
\Vert\nabla_{(AB}\kappa_{CD)}\Vert_{L^2} =\Vert\nabla_{(AB}\kappa_{CD)}\Vert_{H^0_{-3/2}}<\infty.
\]
The result follows. 
\end{proof}

\medskip
\noindent
\textbf{Remark.} Again, let $\kappa_{AB}$ be the solution to equation
\eqref{ApproximateKillingEquation} given by Theorem
\ref{Theorem:ExistenceKS}.  Using the identity
\eqref{IntegrationbyParts} with $\zeta_{ABCD}=\nabla_{(AB}\kappa_{CD)}$
one obtains that
\[
J = \int_{\partial \mathcal{S}_\infty} n^{AB}\kappa^{CD}
\widehat{\nabla_{(AB}\kappa_{CD)}} \mbox{d}S <\infty.
\]
Thus, the invariant $J$ evaluated at the solution $\kappa_{AB}$ given
by Theorem \ref{Theorem:ExistenceKS} can be expressed as a boundary
integral at infinity. A crude estimation of the integrand of the
boundary integral does not allow directly to establish its
boundedness. This follows, however, from Corollary
\ref{Corollary:Boundedness}. 
Hence, the leading order terms of $n_{AB}\kappa_{CD}$ and
$\nabla_{(AB}\kappa_{CD)}$ are orthogonal. 

For an independent proof of this fact, see appendix \ref{AlternativeEstimate}.

\section{The geometric invariant}
\label{Section:Invariant}
In this section we show how to use the functional (\ref{functional})
and the algebraic conditions (\ref{kspd2}) and (\ref{kspd3}) to
construct the desired geometric invariant measuring the deviation of
$(\mathcal{S},h_{ab},K_{ab})$ from Kerr initial data. To this end, let
$\kappa_{AB}$ be a solution to equation
\eqref{ApproximateKillingEquation} as given by Theorem
\ref{Theorem:ExistenceKS}. Furthermore, let $\xi_{AB} \equiv
\tfrac{3}{2}\nabla^{P}{}_{(A}\kappa_{B)P}$. Define
\begin{subequations}
\begin{eqnarray}
&& I_1 \equiv \int_{\mathcal{S}} \Psi_{(ABC}{}^{F}\kappa_{D)F}
\hat{\Psi}^{ABCG}\hat{\kappa}^D{}_G \mbox{d}\mu, \label{I1} \\
&& I_2 \equiv{} \int_{\mathcal{S}} \left(3\kappa_{(A}{}^{E}\nabla_{B}{}^{F}\Psi_{CD)EF}+\Psi_{(ABC}{}^{F}\xi_{D)F}\right) \nonumber \\
&& \hspace{3cm}\times \left(3\hat\kappa^{AP}\widehat{\nabla^{BQ}\Psi^{CD}{}_{PQ}}+\hat\Psi^{ABCP}\hat\xi^D{}_P\right){} \mbox{d}\mu. \label{I2}
\end{eqnarray}
\end{subequations}
The geometric invariant is then defined by
\begin{eqnarray}
I \equiv J + I_1 + I_2. \label{geometric:invariant}
\end{eqnarray}

\medskip
\noindent
\textbf{Remark.} It should be stressed that by construction $I$ is
coordinate independent and that $I\geq 0$. We also have the following
lemma.

\begin{lemma}
The geometric invariant given by \eqref{geometric:invariant} is finite
for an initial data set $(\mathcal{S},h_{ab},K_{ab})$ satisfying the
decay conditions \eqref{BoostedDecay1}-\eqref{BoostedDecay2}.
\end{lemma} 

\begin{proof}
 From Corollary \ref{Corollary:Boundedness} we already have
$J<\infty$. From the form of the decay assumptions
\eqref{BoostedDecay1}-\eqref{BoostedDecay2} we have $\Psi_{ABCD}\in
H^\infty_{-3+\varepsilon}$, $\varepsilon>0$.  By 
Lemma \ref{Lemma:Multiplication} and $\kappa_{AB}\in
H^\infty_{1+\varepsilon}$ we have
\[
\Psi_{(ABC}{}^{F}\kappa_{D)F} \in H^\infty_{-3/2}.
\]
 Thus, again one
finds that $I_1<\infty$.  A similar argument shows that
\[
3\kappa_{(A}{}^{E}\nabla_{B}{}^{F}\Psi_{CD)EF}+\Psi_{(ABC}{}^{F}\xi_{D)F}\in H^\infty_{-3/2},
\]
from where it follows that
 $I_2 <\infty$. Hence, the invariant (\ref{geometric:invariant}) is finite and well defined.
\end{proof}

\medskip
Finally, we are in the position of stating the main result of this
article. It combines all the results in the sections 2 to 7.

\begin{theorem}
Let $(\mathcal{S},h_{ab},K_{ab})$ be an asymptotically Euclidean
initial data set for the Einstein vacuum field equations satisfying on
each of its two asymptotic ends the decay conditions
\eqref{BoostedDecay1}-\eqref{BoostedDecay2} and \eqref{alphabeta} with
a timelike ADM 4-momentum. Furthermore, assume that $\Psi_{ABCD}\neq
0$ and $\Psi_{ABCD}\Psi^{ABCD}\neq 0$ everywhere on $\mathcal{S}$. Let
$I$ be the invariant defined by equations \eqref{functional},
\eqref{I1}, \eqref{I2} and \eqref{geometric:invariant}, where
$\kappa_{AB}$ is given as the only solution to equation
\eqref{ApproximateKillingEquation} with asymptotic behaviour on each
end given by \eqref{KillingSpinorLeading}. The invariant $I$ vanishes
if and only if $(\mathcal{S},h_{ab},K_{ab})$ is locally an initial data set
for the Kerr spacetime.
\end{theorem}

\begin{proof}
Due to our smoothness
assumptions, if $I=0$ it follows that equations
\eqref{kspd1}-\eqref{kspd3} are satisfied on the whole of
$\mathcal{S}$. Thus, the development of $(\mathcal{S},h_{ab},K_{ab})$
will have, at least in a slab, a Killing spinor. Accordingly, it must
be of Petrov type D, N or O on the slab ---see Theorem
\ref{Theorem:TypeDhasalwaysaKS}. The types N and O are excluded by the
assumptions $\Psi_{ABCD}\neq 0$ and $\Psi_{ABCD}\Psi^{ABCD}\neq 0$ on
$\mathcal{S}$ ---by continuity, these conditions will also hold in a
suitably small slab. Thus the development of the data can only be of
Petrov type D ---at least on a suitably small slab.

Now, from the general theory on Killing spinors, we know that
$\xi_{AA'}=\nabla_{A'}{}^Q \kappa_{AQ}$ will be, in general, a complex
Killing vector. In particular, both the real and imaginary parts of
$\xi_{AA'}$ will be real Killing vectors. The Killing initial data for
$\xi_{AA'}$ on $\mathcal{S}$ consists of the fields $\xi$ and
$\xi_{AB}$ on $\mathcal{S}$ calculated from $\kappa_{AB}$ using the
expressions \eqref{xi_sen_1} and \eqref{xi_sen_2}. It can be verified
that
\[
\xi-\hat{\xi}=o_\infty(r^{-1/2}), \quad \xi_{AB}+\hat{\xi}_{AB}=o_\infty(r^{-1/2}).
\]
The latter corresponds to the Killing initial data for the imaginary
part of $\xi_{AA'}$. It follows that the imaginary part of $\xi_{AA'}$
goes to zero at infinity. However, there are no non-trivial Killing
vectors of this type \cite{BeiChr96,ChrOMu81}. Thus, $\xi_{AA'}$ is a
real Killing vector. This means that the spacetime belongs, at least
in a suitably small slab of $\mathcal{S}$, to the generalised Kerr-NUT
class. By construction, it tends to a time translation at infinity so that,
in fact, it is a stationary Killing vector.  By virtue of the decay
assumptions \eqref{BoostedDecay1}-\eqref{BoostedDecay2} the
development of the initial data will be asymptotically flat, and it
can be verified that the Komar mass of each end coincides with the
corresponding ADM mass ---these are non-zero by assumption. Hence,
Theorem \ref{Theorem:SpacetimeCharacterisation} applies and the slab of $\mathcal{S}$ is locally
isometric to the Kerr spacetime.
\end{proof}

\begin{corollary}
If furthermore, the slice $\mathcal{S}$ is assumed, a priori, to
have the same topology as a slice of the Kerr spacetime one has that
the invariant $I$ vanishes if and only if
$(\mathcal{S},h_{ab},K_{ab})$ is an initial data set for the Kerr
spacetime.
\end{corollary}

\begin{proof}
This follows from the  uniqueness of the
maximal globally hyperbolic development of Cauchy data ---see
\cite{ChoGer69}.
\end{proof}

\medskip
\noindent
\textbf{Remark 1.} A improvement of Theorem
\ref{Theorem:SpacetimeCharacterisation} in which no \emph{a priori}
restrictions on the Petrov type of the spacetime are made ---see the
remark after Theorem \ref{Theorem:SpacetimeCharacterisation}--- would
allow to remove the conditions $\Psi_{ABCD}\neq 0$ and
$\Psi_{ABCD}\Psi^{ABCD}\neq 0$, and thus obtain a stronger
characterisation of Kerr data.

\medskip
\noindent
\textbf{Remark 2.} It is of interest to analyse whether the same
conclusion of the corollary can be obtained without making \emph{a
priori} assumptions on the topology of the 3-manifold.

\section{Future prospects}
We have seen that one can construct a geometric invariant for a slice
with two asymptotically flat ends. A natural extension of this work
would be to also allow asymptotically hyperboloidal and asymptotically
cylindrical slices. Furthermore, one would like to analyse parts of
manifolds in the same way. In this case we need to find appropriate
conditions that can be imposed on $\kappa_{AB}$ on the boundary of the
region we would like to study. A typical scenario would be to study
the domain of outer communication for a black hole, or the exterior of
a star.

Another natural question to be asked is how the geometric invariant
behaves under time evolution. A great part of this problem is to
obtain a time evolution of $\kappa_{AB}$ such that it satisfies
\eqref{ApproximateKillingEquation} on every leaf of the foliation.  If
the geometric invariant is small, one could instead use
\eqref{boxkappa} as an approximate evolution equation for the
approximate Killing spinor. In this case the system \eqref{wave1},
\eqref{wave2} could be used to gain control over the evolution.

If some type of constancy or monotonicity property could be
established for the geometric invariant, this would be a useful tool
for studying non-linear stability of the Kerr spacetime and also in
the numerical evolutions of black hole spacetimes.

\section*{Acknowledgements}
We thank A Garc\'{\i}a-Parrado and J M Mart\'{\i}n-Garc\'{\i}a for
their help with computer algebra calculations in the suite {\tt xAct} \cite{xAct}, and M
Mars and N Kamran for valuable comments. TB is funded by a scholarship
of the Wenner-Gren foundations. JAVK is funded by an EPSRC Advanced
Research fellowship.

\appendix

\section{An alternative estimation of the boundary integral}\label{AlternativeEstimate}
In this section we present an alternative argument to show that the
boundary integral 
\[
\int_{\partial \mathcal{S}_r} n^{AB} \kappa^{CD} \widehat{\nabla_{(AB}\kappa_{CD)}} \mbox{d}S,
\]
is finite as $r\rightarrow \infty$ ---cfr. the remark after Corollary
\ref{Corollary:Boundedness}. For simplicity, we only consider the
non-boosted case, so we have
\[
\kappa_{AB}=\pm \frac{\sqrt{2}}{3}r n_{AB} +O(1).
\]
A similar, but much lengthier argument can be implemented in the
boosted case. It is only necessary to consider the finiteness of the integral
\begin{equation}
\label{BoundaryIntegral}
r \int_{\partial \mathcal{S}_r} n^{AB} n^{CD} \widehat{\nabla_{(AB}\kappa_{CD)}} \mbox{d}S \quad \mbox{ as } r\rightarrow \infty.
\end{equation}

\medskip
We begin by investigating the multipole structure of $\xi_{ABCD}\equiv
\nabla_{(AB}\kappa_{CD)}$ in an asymptotically flat end
$\mathcal{U}\subset\mathcal{S} $.  The equation satisfied by
$\xi_{ABCD}$ is
\begin{equation}
\label{Divergence:xi}
\nabla^{AB} \xi_{ABCD} - 2\Omega^{ABF}{}_{(C} \xi_{D)ABF}=0,
\end{equation}
---see equation \eqref{ApproximateKillingEquation}.  As
$\mathcal{U}\approx (r_0,\infty)\times \Sphere^2$, with $r_0\in\Real$,
it will be convenient to work in spherical coordinates. For
simplicity, we adopt the point of view that all the angular dependence
of the various functions involved is expressed in terms of
(spin-weighted) spherical harmonics. Accordingly, we use the
differential operators $\eth, \;\bar{\eth} \in
\mbox{T}\Sphere^2$---see e.g. \cite{PenRin84}. Let
$\omega_+,\;\omega_-\in \mbox{T}^*\Sphere^2$ denote the 1-forms dual
to $\eth$ and $\bar{\eth}$: 
\[
\langle \eth, \omega_+ \rangle =1, \quad \langle \bar{\eth},
\omega_-\rangle =1.
\]
In addition, we consider
$\partial_r\in \mbox{T}\mathcal{U}$. The operators $\eth,\;\bar{\eth}$
are extended into $\mbox{T} \mathcal{U}$ by requiring that
\[
[\eth,\partial_r]=[\bar{\eth},\partial_r]=0.
\]
Again, let $\mbox{d}r\in\mbox{T}^*\mathcal{U}$ denote the form dual to $\partial_r$. One has
that
\[
\delta_{ij} \mbox{d} x^i \otimes \mbox{d}x^j = \mbox{d}r \otimes \mbox{d} r
+ r^2 \left( \omega_+ \otimes \omega_- + \omega_-\otimes \omega_+ \right).
\]
Now, recalling that
\[
h_{ij} = -\left(1+\frac{2m}{r}\right) \delta_{ij} + o_{\infty}(r^{-3/2}),
\]
we introduce the following frame and coframe:
\begin{eqnarray*}
&& e_{01}= \left( 1-\frac{m}{r} \right)\partial_r +
o_{\infty}(r^{-3/2}), \quad \sigma^{01}= \left( 1+\frac{m}{r} \right)\mbox{d}r +
o_{\infty}(r^{-3/2})  \\
&& e_{00}= \left( 1-\frac{m}{r} \right)\frac{1}{r}\eth +
o_{\infty}(r^{-5/2}), \quad \sigma^{00} =  \left( 1+\frac{m}{r}
\right)r \omega_+ + o_\infty(r^{-1/2})\\
&& e_{11}= \left( 1-\frac{m}{r} \right)\frac{1}{r}\bar{\eth} +
o_{\infty}(r^{-5/2}), \quad \sigma^{11}= \left( 1+\frac{m}{r} \right)r
\omega_- + o_\infty(r^{-1/2}).
\end{eqnarray*}
The fields $e_{AB}$ and $\sigma^{AB}$ satisfy
\[
\langle e_{AB} , \sigma^{CD} \rangle = h_{AB}{}^{CD}, \quad h =
h_{ABCD} \sigma^{AB} \otimes \sigma^{CD}.
\]
where $h_{ABCD} \equiv - \epsilon_{A(C}\epsilon_{D)B}$. Let $\mu_{AB}$
denote a smooth spinorial field. Its covariant
derivative $D_{EF}\mu_{AB}$ can be computed using
\[
D_{EF}\mu_{AB}= e_{EF}(\mu_{AB}) -\Gamma_{EF}{}^Q{}_A \mu_{QB}-\Gamma_{EF}{}^Q{}_B \mu_{AQ},
\]
where $\Gamma_{EF}{}^Q{}_A$ denote the spin coefficients of the frame
$e_{AB}$.

\medskip
The components of the spinor field $\xi_{ABCD}$ with respect to the frame $e_{AB}$  can be written as
\[
\xi_{ABCD} = \xi_0 \epsilon^0_{ABCD}+ \xi_1 \epsilon^1_{ABCD} +\xi_2
\epsilon^2_{ABCD}+ \xi_3 \epsilon^3_{ABCD}+ \xi_4 \epsilon^4_{ABCD}, 
\]
where 
\[
\epsilon^k_{ABCD} \equiv \epsilon_{(A}{}^{(E} \epsilon_B{}^F \epsilon_C{}^G \epsilon_{D)}{}^{H)_k},
\]
where ${}^{(EFGH)_k}$ means that after symmetrisation, $k$ indices are set
to $1$.  In terms of this formalism, equation \eqref{Divergence:xi} is
given by
\begin{equation}
\label{DivergenceFrame}
\epsilon^{AP}\epsilon^{BQ} e_{PQ}(\xi_{ABCD}) - 4 \Gamma^{ABQ}{}_{(A}
\xi_{BCD)Q} + 2 K^{ABQ}{}_{(A} \xi_{BCD)Q} - 2\Omega^{ABQ}{}_{(C}\xi_{D)ABQ}=0.
\end{equation}
Recalling that by assumption $\xi_{ABCD}=o_\infty(r^{-3/2})$, a
lengthy but straightforward calculation shows that
\eqref{DivergenceFrame} implies the equations
\begin{subequations}
\begin{eqnarray}
&& \partial_r \xi_1 -\frac{1}{r}\bar{\eth}\xi_0 + \frac{1}{6}\frac{1}{r} \eth \xi_2 + \frac{3}{r}\left( 1+\frac{m}{r} \right) \xi_1 = o_\infty(r^{-5}), \label{Frame1}\\
&&\partial_r \xi_2 + \frac{3}{2}\frac{1}{r} \bar{\eth} \xi_1 +
\frac{3}{2}\frac{1}{r} \eth\xi_3 +
\frac{3}{r}\left( 1+\frac{m}{r} \right) \xi_2 = o_\infty(r^{-5}), \label{Frame2}\\
&& \partial_r \xi_3 + \frac{1}{r} \eth \xi_4 - \frac{1}{6}\frac{1}{r} \bar{\eth} \xi_2 + \frac{3}{r}\left( 1+\frac{m}{r} \right)\xi_3 =
o_\infty(r^{-5}). \label{Frame3}
\end{eqnarray}
\end{subequations}

\medskip
A computation shows that
\[
n_{(AB} n_{CD)}= \epsilon^2_{ABCD},
\]
so that the boundary integral \eqref{BoundaryIntegral} involves only
the component $\xi_2$. Furthermore, only the harmonic $Y_{0,0}$
(monopole) contributes to the integral as $\epsilon^2_{ABCD}$ is a
constant spinor in our frame. From the equations
\eqref{Frame1}-\eqref{Frame3}, it follows that the coefficient
$\xi_{2;0}$ of $\xi_2$ associated to the harmonic $Y_{0,0}$ satisfies
the ordinary differential equation 
\[
\left(1-\frac{m}{r}\right)\partial_r \xi_{2;0} + \frac{3}{r} \xi_{2;0}
=f(r), \quad f(r)=o_\infty(r^{-5}). 
\]
Consequently one has that
\[
\xi_{2;0} = \frac{\alpha}{(r-m)^3} + \frac{1}{(r-m)^3}\int r(r-m)^2
f(r) \mbox{d}r, \quad \alpha\in \Complex.
\]
It follows that
\[
\xi_{2;0} = \frac{\alpha}{r^3} +o_\infty(r^{-4}).
\]
Using this last expression in the integral \eqref{BoundaryIntegral}
and recalling that $\mbox{d}S=O(r^2)$, it follows that 
\[
r \int_{\partial \mathcal{S}_r} n^{AB} n^{CD}
\widehat{\nabla_{(AB}\kappa_{CD)}} \mbox{d}S = 4\pi\alpha<\infty.
\]
It is worth noting that the constant $\alpha$ contains information of
global nature and it is only known after one has solved the
approximate Killing spinor equation.

\section{Tensor expressions}
For many applications, it is useful to have tensor expressions for the invariants. To this end, define the following tensors on $\mathcal{S}$:
\begin{align*}
\kappa_a&\equiv \sigma_a{}^{AB}\kappa_{AB},&
\zeta&\equiv \xi,\\
\zeta_a&\equiv \sigma_a{}^{AB}\xi_{AB},&
\zeta_{ab} &\equiv \sigma_a{}^{AB}\sigma_b{}^{CD}\xi_{ABCD},\\
C_{ac} &\equiv E_{ac}+\mbox{i}B_{ac}.
\end{align*}
Here $\epsilon_{abc}$, $E_{ac}$ and $B_{ac}$ are the pull-backs of 
$\tfrac{1}{\sqrt{2}}\tau^\mu\epsilon_{\mu\alpha\beta\gamma}$, 
$\tfrac{1}{2}\tau^\gamma \tau^\delta C_{\alpha\beta\gamma\delta}$ 
and $\tfrac{1}{4}\epsilon_{\mu\nu\gamma\delta}\tau^\beta \tau^\delta C_{\alpha\beta}{}^{\mu\nu}$ respectively.
Observe that we are using a negative definite metric.
In this section we assume $K_{ab}=K_{ba}$.

The tensorial versions of the equations \eqref{xi_sen_1}, \eqref{xi_sen_2}, \eqref{xi_sen_3} then read
\begin{subequations}
\begin{align*}
\zeta &=D^a\kappa_a,\\
\zeta_a &=\tfrac{3}{2\sqrt{2}}\mbox{i}\epsilon_{abc} D^c\kappa^b
-\tfrac{3}{4}K_{ab}\kappa^b+\tfrac{3}{4}K_b{}^b\kappa_a,\\
\zeta_{ab}&=D_{(a}\kappa_{b)}-\tfrac{1}{3} h_{ab}D^c\kappa_c-
\tfrac{1}{\sqrt{2}}\mbox{i}\epsilon_{cd(a}K_{b)}{}^d \kappa^c.
\end{align*}
\end{subequations}
Note that the spatial Killing spinor equation $\zeta_{ab}=0$ reduces to the conformal Killing vector equation in the time symmetric case ($K_{ab}=0$).

Expressed in terms of these tensors the elliptic equation \eqref{ApproximateKillingEquation} takes the form
\begin{equation}\label{elliptictensor}
D^b\zeta_{ab}-\tfrac{1}{\sqrt{2}}\mbox{i}\epsilon_{acd}K^{bc}\zeta_b{}^d=0.
\end{equation}

Let $\kappa_a\in H^\infty_{3/2}$ be the solution to \eqref{elliptictensor} with the asymptotics
\begin{align*}
\kappa_i ={}& 
\mp\frac{\sqrt{2}E}{3m}\left (1+\frac{2E}{r}\right)x_i 
\pm\frac{2\mbox{i}}{3m}\left(1
+\frac{4E}{r}
-\frac{m^2+2(n\cdot p)^2}
{\sqrt{m^2+(n\cdot p)^2}r}
\right)\epsilon_i{}^{jk}p_j x_k
+o_\infty(r^{-1/2}),
\end{align*}
at each end, where $p_\mu=(E, p_i)$ is the ADM-4 momentum, $m\equiv\sqrt{p^\mu p_\mu}$, and $n \cdot p=r^{-1}x^i p_i$. The metric and extrinsic curvature are assumed to have the asymptotics \eqref{BoostedDecay1} and \eqref{BoostedDecay2} respectively.

The integrand in \eqref{functional} is 
\[
\mathfrak{J}\equiv\xi_{ABCD}\hat\xi^{ABCD}=\zeta_{ab}\bar\zeta^{ab}.
\]

 From the equation
\begin{align*}
\sigma_a{}^{AB}\sigma_b{}^{CD}\Psi_{(ABC}{}^F\kappa_{D)F}&=
\tfrac{1}{\sqrt{2}}\mbox{i}\epsilon_{cd(a}C_{b)}{}^d\kappa^c.
\end{align*}
we get the integrand for the $I_1$ part of the invariant
\[
\mathfrak{I}_1\equiv\Psi_{(ABC}{}^F\kappa_{D)F}\hat\Psi^{ABCP}\hat\kappa^D{}_{P}=
-\tfrac{1}{2}C^{bc}\bar C_{bc}\kappa^a\bar\kappa_a
+\tfrac{1}{2}C_b{}^c\bar C_{ac}\kappa^a\bar\kappa^b
+\tfrac{1}{4}C_a{}^c\bar C_{bc}\kappa^a\bar\kappa^b.
\]

In order to discuss the integrand of $I_2$ we introduce the spinor
$\Sigma_{ABCD}\equiv \nabla_{(A}{}^F\Psi_{BCD)F}$, and its tensor
equivalent
$\Sigma_{ab}=\sigma_a{}^{AB}\sigma_b{}^{CD}\Sigma_{ABCD}$. One finds that
\begin{align*}
0&=\sigma_a{}^{AB}\nabla^{CD}\Psi_{ABCD}=D^bC_{ab}-\tfrac{\mbox{i}}{\sqrt{2}}\epsilon_{acd}C^{bc}K_b{}^d, \\
\Sigma_{ab}&=\tfrac{\mbox{i}}{\sqrt{2}}\epsilon_{df(a}D^fC^d{}_{b)}
+\tfrac{1}{2}C^{cd}K_{cd}h_{ab}+C_{ab}K^f{}_f
-\tfrac{3}{2}C^c{}_{(a}K_{b)c}.
\end{align*}
The integrand for $I_2$ is given by 
\begin{align*}
\mathfrak{I}_2={}&(3\kappa_{(A}{}^F\Sigma_{BCD)F}+\Psi_{(ABC}{}^F\xi_{D)F})
(3\hat\kappa^{AP}\hat\Sigma^{BCD}{}_P+\hat\Psi^{ABCP}\hat\xi^D{}_P)
\nonumber\\
={}&-\tfrac{9}{2}\Sigma^{bc}\bar\Sigma_{bc}\kappa^a\bar\kappa_a
+\tfrac{9}{2}\Sigma_b{}^c\bar\Sigma_{ac}\kappa^a\bar\kappa^b
+\tfrac{9}{4}\Sigma_a{}^c\bar\Sigma_{bc}\kappa^a\bar\kappa^b
+\tfrac{3}{2}\bar\Sigma_{bc} C^{bc}\bar\kappa^a\zeta_a
-\tfrac{3}{4}\bar\Sigma_{ac} C_b{}^c\bar\kappa^a\zeta^b \nonumber \\
&-\tfrac{3}{2}\bar\Sigma_{bc} C_a{}^c\bar\kappa^a\zeta^b
+\tfrac{3}{2}\Sigma_{bc}\bar C^{bc}\kappa^a\bar\zeta_a
-\tfrac{3}{4}\Sigma_{ac}\bar C_b{}^c\kappa^a\bar\zeta^b
-\tfrac{3}{2}\Sigma_{bc}\bar C_a{}^c\kappa^a\bar\zeta^b
+\tfrac{1}{2}C^{bc}\bar C_{bc}\zeta^a\bar\zeta_a \nonumber \\
&+\tfrac{1}{2}C_b{}^c\bar C_{ac}\zeta^a\bar\zeta^b
+\tfrac{1}{4}C_a{}^c\bar C_{bc}\zeta^a\bar\zeta^b.
\end{align*}
The complete invariant is given by
\[
I=\int_\mathcal{S}(\mathfrak{J}+\mathfrak{I}_1+\mathfrak{I}_2)\mbox{d}\mu.
\]


\end{document}